\documentclass[11pt]{article}

\usepackage[round]{natbib}
\usepackage{amsthm}
\usepackage{amssymb}
\usepackage{graphicx}
\usepackage{amsmath}
\usepackage{verbatim}
\usepackage{setspace}
\usepackage[margin=1.00in]{geometry}
\usepackage{changepage}
\usepackage{url}
\usepackage{booktabs}
\usepackage{float}

\usepackage{mathtools}
\usepackage{dsfont}

\usepackage[title]{appendix}

\usepackage{hyperref}
\hypersetup{colorlinks, linkcolor=blue, citecolor=blue}

\newtheorem{assumption}{Assumption}
\newtheorem{proposition}{Proposition}

\newtheorem{lemma}{Lemma}

\usepackage{enumitem}
\usepackage{cleveref}

\newlist{assumptionitems}{enumerate}{1}
\setlist[assumptionitems]{label=\arabic*., ref=\theassumption.\arabic*, align=left, leftmargin=2.2em, topsep=0pt, partopsep=0pt}
\Crefname{assumption}{Assumption}{Assumptions}
\crefname{assumptionitemsi}{assumption}{assumptions}
\Crefname{assumptionitemsi}{Assumption}{Assumptions}

\def\bs{\boldsymbol}
\def\t{^{\top}}

\begin{document}

\title{\textbf{Testing for Unobserved Heterogeneity in Censored Duration Models: EM Approach}}
\author{Hiroyuki Kasahara \\
Vancouver School of Economics\\
University of British Columbia \\
hkasahar@mail.ubc.ca
\and
Hirokazu Matsuyama\\
DENTSU SOKEN INC.\\
matsuyama.hirokazu@dentsusoken.com
\and
Katsumi Shimotsu\thanks{Corresponding author.}\\
Faculty of Economics \\
University of Tokyo\\
shimotsu@e.u-tokyo.ac.jp
\and
Shota Takeishi\\
Department of Statistics and Data Science\\
Washington University in St.~Louis\\
shota@wustl.edu}
\date{September 29, 2026}

\maketitle

\begin{abstract}

Ignoring unobserved heterogeneity in duration models biases parameter estimates and invalidates inference, but testing for it is non-regular: the null hypothesis lies on the boundary of the parameter space and some parameters are unidentified under the null. These features render standard asymptotic theory inapplicable. This paper develops an EM test for unobserved heterogeneity in censored Weibull duration models, building on the EM approach of \citet*{lcm09bm}. The test statistic has an asymptotic null distribution equal to the square of $\max\{0, N(0,1)\}$, hence critical values require neither simulation nor bootstrap, and the test accommodates covariate-dependent censoring of arbitrary form. Monte Carlo simulations compare the EM test with the likelihood ratio test (LRT) of \citet{chowhite10joe}, information matrix tests, and Lagrange multiplier tests. The EM test has empirical size close to the nominal level for sample sizes of 500 or more, where the LRT remains markedly conservative and the other tests over-reject. Its size-adjusted power is comparable to that of the LRT and higher than that of the other tests. Because size adjustment requires knowledge of the data-generating process and is unavailable in practice, the EM test attains higher power than the LRT in most designs as the tests would actually be applied. In an application to the Stanford Heart Transplant data, the EM test rejects homogeneity in every covariate specification, whereas the LRT's conclusion depends on a user-chosen set of admissible parameter values and on the specification.
\end{abstract}



\section{Introduction}

Weibull duration models have been widely used in economics to analyze transitions, such as unemployment spells \citep{Lancaster1979}, plant survival \citep{Chen02ijio}, firm survival \citep{Kalnins13aej}, workers' compensation claims \citep{Butler01restat}, school-to-work transitions \citep{Pastore21labour}, and waiting times between stock transactions \citep{Engle98em}. \citet{Kiefer88jel}, \citet{Lancaster90}, and \citet{vandenberg01hdbk} review the theory and applications of duration models.

As demonstrated by \citet{heckmansinger84em}, parameter estimates in duration models can be substantially affected by assumptions concerning unobserved heterogeneity, and failing to account for it may lead to bias. Consequently, several tests for unobserved heterogeneity have been developed. \citet{Chesher1984} and \citet{Lancaster1984, Lancaster1985} develop information matrix (IM) tests based on the fact that neglected unobserved heterogeneity violates the information matrix equality of the homogeneous model. \citet{Kiefer1985}, \citet{Sharma1987}, and \citet{Prieger2003} propose Lagrange multiplier (LM) tests that exploit the Laguerre polynomial representations of the exponential and Weibull distributions. A parallel literature in biostatistics develops score tests for homogeneity across groups in survival data \citep{commengesandersen95lida, bolfarinevalenca05bj}.

\citet{chowhite10joe} develop a likelihood ratio test (LRT) to detect unobserved heterogeneity in exponential and Weibull duration models using a finite mixture framework. Their procedure tests the null hypothesis of no unobserved heterogeneity against a two-component finite mixture alternative. The LRT statistic has a nonstandard asymptotic distribution due to the presence of unidentified parameters and boundary issues under the null hypothesis. Their simulations show that the test has power against a broad range of heterogeneity distributions and better size and power properties than the IM and LM tests.

Many duration datasets contain censored observations due to survey-based sampling designs. For example, in unemployment duration studies, some data are obtained from individuals who are currently unemployed, rather than from tracking individuals over the full course of their unemployment spells. In such cases, only part of the spell is observed. Moreover, the censoring time may depend on individual characteristics such as age and education.

In censored duration models, existing tests for unobserved heterogeneity face challenges related to both size control and implementation. First, the LRT of \citet{chowhite10joe}, the IM test, and the LM tests suffer from finite-sample size distortions: the IM and LM tests are typically severely oversized, while the LRT tends to be undersized. Second, the LRT is difficult to apply under covariate-dependent censoring because its asymptotic null distribution is known only for the cases of fixed censoring and random (independent) censoring. Third, implementing the LRT is computationally intensive, as its asymptotic null distribution is a functional of a Gaussian process and is approximated by the weighted bootstrap of \citet{hansen96em}. Finally, the LRT requires the user to specify a set of admissible values for the mixture components, and its asymptotic null distribution depends on that set. In finite samples, the size and the power of the LRT, and hence its conclusions, are sensitive to that choice.

This paper develops a new test for unobserved heterogeneity in censored Weibull duration models that is straightforward to implement and exhibits good finite-sample performance. The proposed test, referred to as the EM test, adopts the finite mixture framework of \citet{chowhite10joe} and employs the EM approach \citep{lcm09bm, chenli09as, kasaharashimotsu15jasa} to test the null hypothesis of no unobserved heterogeneity against a two-component finite mixture alternative.

The proposed test offers several attractive features. First, the asymptotic null distribution of the test statistic is the square of $\max\{0, N(0,1)\}$. Critical values therefore require neither simulation nor the bootstrap. Second, the test accommodates covariate-dependent censoring of arbitrary form without requiring special adjustments, provided that the duration and the censoring time are conditionally independent given the covariates.
Third, the test demonstrates good finite-sample performance: its empirical size is close to the nominal level when the sample size is 500 or more, under both covariate-independent and covariate-dependent censoring. Moreover, it exhibits better size control than the LRT, IM, and LM tests. Fourth, for sample sizes of 500 or more, the proposed test achieves size-adjusted power higher than that of the IM and LM tests and comparable to that of the LRT. Size adjustment is unavailable in practice, since it requires knowledge of the data-generating process, including the censoring mechanism. Without size adjustment, the EM test attains higher power than the LRT in most designs because the LRT is markedly conservative. Furthermore, an application to the Stanford Heart Transplant data illustrates the practical difference: the EM test rejects homogeneity at the 1\% level in all five covariate specifications, whereas the LRT fails to reject at the 5\% level under its narrowest admissible set in every specification and rejects in most of them under wider admissible sets.

Testing and inference in finite mixture models have attracted considerable interest in econometrics. \citet{chenponomarevatamer14joe} study inference on the mixing probability in a two-component mixture and show that the profiled likelihood ratio statistic has asymptotic limits that vary discontinuously near the region where the mixture is not identified; they use a parametric bootstrap to restore uniform coverage. \citet{gukoenkervolgushev18et} develop an LRT based on a nonparametric maximum likelihood estimator of the mixing distribution. They derive its asymptotic distribution, but find that the asymptotic critical values are unlikely to control size and instead advocate a parametric bootstrap, whose consistency they establish.

A related literature relaxes the parametric and censoring assumptions of the Weibull model. \citet{KhanTamer07joe} develop a rank-based estimator of the regression coefficients in nonparametric censored duration models that allow for general forms of censoring. \citet{KhanTamer09joe} propose an instrumental variables estimator that allows for covariate-dependent censoring and endogenous censoring. \citet{Szydlowski19jae} and \citet{Sakaguchi24jae} analyze the set identification of parameters in duration models with endogenous censoring. \citet{Han1990} and \citet{Meyer1990} combine flexible hazard functions with parametric assumptions about unobserved heterogeneity.

The remainder of this paper is organized as follows. Section \ref{sec:Weibull_model} introduces the Weibull duration model with censoring. Section \ref{sec:tests} formulates unobserved heterogeneity as a two-point mixture, develops the EM test and derives its asymptotic null distribution, and reviews the LRT of \citet{chowhite10joe}. Section \ref{sec:MC} reports Monte Carlo simulation results under covariate-independent and covariate-dependent censoring. Section \ref{sec:realdata} applies the tests to data from the Stanford Heart Transplant Program. Section \ref{sec:conclusion} concludes. Appendices \ref{app:assumptions}--\ref{app:aux} collect the assumptions, proofs, and auxiliary results, respectively.

We use the following notation throughout. All limits below are taken as $n \rightarrow \infty$, unless stated otherwise. We use boldfaced letters to denote vectors and matrices. Let $\coloneq $ denote ``equals by definition.'' Let $\mathbb{R}^+ \coloneq (0,\infty)$. Let $\mathds{1}\{A\}$ denote an indicator function that takes the value of one when $A$ is true and zero otherwise. For a $k\times 1$ vector $\bs{a}$, let $\|\bs{a}\|$ and $\|\bs{a} \|_\infty$ denote the Euclidean norm and the uniform norm, respectively. For a $k\times 1$ vector $\bs{a}$ and a function $f(\bs{a})$, let $\nabla_{\bs{a}}f(\bs{a})$ denote the $k\times 1$ vector of the derivative $(\partial/ \partial {\bs{a}}) f(\bs{a})$, and let $\nabla_{\bs{a}\bs{a}^{\top}}f(\bs{a})$ denote the $k\times k$ matrix of the derivative $(\partial^2/ \partial{\bs{a}}\partial{\bs{a}}^{\top}) f(\bs{a})$.

\section{Weibull duration model with censoring}\label{sec:Weibull_model}

This section introduces the Weibull duration model with censoring. Our notation largely follows that of \citet{chowhite10joe}, except that we write $\varphi$ for their regression function $g$. Let $\{(Y_t, C_t, \bs{X}_t)\}$ be a strictly stationary geometric $\beta$-mixing stochastic process, where $Y_t$ and $C_t$ are scalar-valued and strictly positive, and $\bs{X}_t$ is $\mathbb{R}^k$-valued. The $\beta$-mixing condition allows $\bs{X}_t$ to depend on the entire history of $\{Y_s\}_{s<t}$, thus accommodating autoregressive conditional duration models. Here, $Y_t$ denotes the duration, and $C_t$ denotes the censoring variable. We allow for covariate-dependent censoring, meaning that $C_t$ may be arbitrarily correlated with the covariates $\bs{X}_t$. Under right censoring, the observed duration is $Y_t^{c}\coloneq \min\{Y_t, C_t\}$. Let $D_t \coloneq \mathds{1}\{Y_t^{c} = Y_t\}$, which equals 1 if the observed duration is not censored. In this setting, we observe $\{(Y_t^c, D_t,\bs{X}_t)\}_{t=1}^n$, although $C_t$ is not observed when $D_t = 1$.

Following \citet{KhanTamer07joe}, we assume throughout that $Y_t$ and $C_t$ are conditionally independent given $\bs{X}_t$.\footnote{For dependent observations, Assumption \ref{assn1}(ii) in the Appendix additionally requires the joint conditional law of $(Y_t,C_t)$ given $\bs{X}_t$ and the full past to depend only on $\bs{X}_t$.} This assumption rules out endogenous censoring, in which $C_t$ is related to the duration itself beyond what $\bs{X}_t$ captures.

We assume that the uncensored duration $Y_t$ conditional on $\bs{X}_t$ follows a Weibull distribution with the following parameterization
\begin{equation}\label{eq:wblpdf}
f(y| \bs{x}; \delta^*, \bs{\beta}^*, \gamma^*) = \delta^* \gamma^* \varphi(\bs{x}; \bs{\beta}^*)y^{\gamma^*-1}\exp(-\delta^* \varphi(\bs{x};\bs{\beta}^*) y^{\gamma^*}) ,
\end{equation}
where $(\delta^*,\bs{\beta}^*,\gamma^*) \in D \times B \times \Gamma \subset \mathbb{R}^+ \times \mathbb{R}^d \times \mathbb{R}^+$, and $\varphi(\cdot;\bs{\beta})$ is strictly positive for every $\bs{\beta} \in B$. We assume that $\bs{X}_t$ does not include a constant term, because an intercept in $\bs{\beta}^*$ would not be separately identified from $\delta^*$: in the widely used specification $\varphi(\bs{x}; \bs{\beta}^*) = \exp(\bs{x}\t\bs{\beta}^*)$, we have $\delta^* \varphi(\bs{x}; \bs{\beta}^*) = \exp(\log \delta^* + \bs{x}\t\bs{\beta}^*)$, so that $\log \delta^*$ serves as the intercept.

We model the joint conditional density of $(Y_t^{c}, D_t)$ given $\bs{X}_t$ as
\begin{equation}\label{eq:model1} 
\begin{aligned}
f^{c}(y^{c},d | \bs{x};\delta^*,\bs{\beta}^*,\gamma^*) \coloneq & \left\{f(y^{c} | \bs{x};\delta^*,\bs{\beta}^*,\gamma^*)\right\}^{d}\left\{1-F(y^{c} | \bs{x};\delta^*,\bs{\beta}^*,\gamma^*)\right\}^{1-d} \\
& \times \left\{1-G(y^{c}|\bs{x})\right\}^{d}\left\{g(y^{c}|\bs{x})\right\}^{1-d},
\end{aligned} 
\end{equation}
where $f(\cdot| \bs{x}; \delta^*, \bs{\beta}^*, \gamma^*)$ is defined by \eqref{eq:wblpdf}, $F(\cdot| \bs{x};\delta^*,\bs{\beta}^*,\gamma^*)$ is the distribution function of $Y_t$ given $\bs{X}_t$, and $G(\cdot|\bs{x})$ is the distribution function of $C_t$ given $\bs{X}_t$. Here $g(\cdot|\bs{x})$ denotes the density of $G(\cdot|\bs{x})$ with respect to a $\sigma$-finite measure $\nu_{\bs{x}}$ that dominates the conditional distribution of $C_t$ given $\bs{X}_t = \bs{x}$, so that $(Y_t, C_t)$ has a conditional density with respect to the product of Lebesgue measure and $\nu_{\bs{x}}$.\footnote{Correspondingly, \eqref{eq:model1} is a density of $(Y_t^c, D_t)$ with respect to the measure that equals Lebesgue measure on $\{d = 1\}$ and $\nu_{\bs{x}}$ on $\{d = 0\}$, because $Y_t^c = Y_t$ when $D_t = 1$ and $Y_t^c = C_t$ when $D_t = 0$. When $C_t$ is continuously distributed, $\nu_{\bs{x}}$ is Lebesgue measure and $g(\cdot|\bs{x})$ is the ordinary conditional density; when the conditional distribution of $C_t$ is discrete or degenerate, $\nu_{\bs{x}}$ is counting measure on its support.} This formulation imposes no restriction on the form of $G(\cdot|\bs{x})$ beyond the conditional independence of $Y_t$ and $C_t$ given $\bs{X}_t$. For maximum likelihood estimation in this model, we can drop the terms $\{1-G(y^{c}|\bs{x})\}^{d}$ and $\{g(y^{c}|\bs{x})\}^{1-d}$ from the likelihood function because they enter multiplicatively and do not depend on the model parameters.

Model \eqref{eq:model1} encompasses random censoring and fixed (Type I) censoring as special cases. In the case of random censoring, $C_t$ is independent of $(Y_t,\bs{X}_t)$. Consequently, model \eqref{eq:model1} simplifies to 
\begin{align}\label{eq:model_random}
f^{cr}(y^{c},d | \bs{x};\delta^*,\bs{\beta}^*,\gamma^*) \coloneq & \left\{f(y^{c} | \bs{x};\delta^*,\bs{\beta}^*,\gamma^*)\right\}^{d}\left\{1-F(y^{c} | \bs{x};\delta^*,\bs{\beta}^*,\gamma^*)\right\}^{1-d}\notag \\
	& \times \left\{1-G(y^{c})\right\}^{d}\left\{g(y^{c})\right\}^{1-d}.
\end{align}
In the case of fixed censoring, $C_t=c<\infty$ with probability one. Then, model \eqref{eq:model1} reduces to 
\begin{align}\label{eq:model_fixed}
f^{cf}(y^{c},d | \bs{x};\delta^*,\bs{\beta}^*,\gamma^*) \coloneq \left\{f(y^{c} | \bs{x};\delta^*,\bs{\beta}^*,\gamma^*)\right\}^{d}\left\{1-F(c| \bs{x};\delta^*,\bs{\beta}^*,\gamma^*)\right\}^{1-d}
\end{align}
on the support $\{(y^c,1): 0 < y^c < c\} \cup \{(c,0)\}$: uncensored observations satisfy $0 < y^c < c$, and censored ones form an atom at $(c,0)$ with probability $1-F(c|\bs{x};\delta^*,\bs{\beta}^*,\gamma^*)$.

\section{Tests for unobserved heterogeneity}\label{sec:tests}

This section develops a test for the presence of unobserved heterogeneity in $\delta^*$ based on a finite mixture model. 

\subsection{Modeling unobserved heterogeneity using a mixture model}

Following \citet{heckmansinger84em} and \citet{chowhite10joe}, we model unobserved heterogeneity in $\delta^*$ using a two-point discrete distribution with point masses at $\delta_{1}^*$ and $\delta_{2}^*$. The discrete mixture specification has been widely adopted in the literature; see, for example, \citet{Gritz1993}, \citet{Ferrall1997}, and \citet{vandenberg98em}. As argued by \citet{chowhite10joe}, finite mixture models offer a flexible approximation to the unknown heterogeneity distribution and facilitate tests capable of detecting a broad range of unobserved heterogeneity. A commonly used alternative is the gamma mixture model, which assumes that $\delta^*$ follows a gamma distribution and offers an analytically convenient framework for incorporating unobserved heterogeneity into duration models \citep{Lancaster1979, Meyer1990, Han1990}. However, this approach has been criticized by \citet{heckmansinger84em} as being somewhat ad hoc. In particular, the gamma distribution may fail to capture key features of the true heterogeneity distribution, such as multimodality or skewness. Another alternative is to leave the heterogeneity distribution unspecified and estimate it by nonparametric maximum likelihood \citep{gukoenkervolgushev18et}.

With a two-point distribution of $\delta^*$, the joint conditional density of $(Y_t^{c}, D_t)$ given $\bs{X}_t$ is written as a mixture of Weibull densities
\begin{align}\label{eq:model_mixture}
f^{c}_{a}(y^{c},d | \bs{x};\pi^*,\delta_{1}^*, \delta_{2}^*,\bs{\beta}^*,\gamma^*) \coloneq & \pi^* f^{c}(y^{c},d | \bs{x};\delta_1^*,\bs{\beta}^*,\gamma^*) + (1-\pi^*)f^{c}(y^{c},d | \bs{x};\delta_2^*,\bs{\beta}^*,\gamma^*) ,
\end{align}
where $f^{c}(y^{c},d | \bs{x};\delta^*,\bs{\beta}^*,\gamma^*)$ is defined in \eqref{eq:model1} and $(\pi^*,\delta_{1}^*,\delta_{2}^*,\bs{\beta}^*,\gamma^*) \in [0,1] \times D\times D \times B \times \Gamma \subset [0,1] \times \mathbb{R}^{+} \times \mathbb{R}^{+}\times \mathbb{R}^{d} \times \mathbb{R}^{+}$. In this model, unobserved heterogeneity is absent when $\delta_{1}^* = \delta_{2}^*$ or $\pi^*(1-\pi^*) = 0$. For maximum likelihood estimation of the parameters in this model, we can drop the terms $\{1-G(y^{c}|\bs{x})\}^d$ and $\{g(y^{c}|\bs{x})\}^{1-d}$ in $f^{c}(y^{c},d | \bs{x};\delta,\bs{\beta},\gamma)$, because these terms enter multiplicatively in both $f^{c}(y^{c},d | \bs{x}; \delta_{1}, \bs{\beta}, \gamma)$ and $f^{c}(y^{c},d | \bs{x}; \delta_{2}, \bs{\beta}, \gamma)$ and $\pi + (1-\pi) = 1$.

Using model \eqref{eq:model_mixture}, we can express the null hypothesis of no unobserved heterogeneity as the parameter restrictions
\begin{equation} \label{eq:null_hypothesis}
H_0: \delta_{1}^* = \delta_{2}^*\ \text{ or }\ \pi^*(1-\pi^*) = 0\ \text{ against }\ H_1: \pi^* \in (0,1) \text{ and } \delta_{1}^* \neq \delta_{2}^*.
\end{equation}

\subsection{EM test for unobserved heterogeneity}

In this section, we develop an EM test for unobserved heterogeneity, building on the EM approach of \citet{lcm09bm}. The proposed test is primarily motivated by the limitations of the LRT when applied to testing $H_0$ against $H_1$ in equation \eqref{eq:null_hypothesis}. Under the null hypothesis, some parameters are not identified: in model \eqref{eq:model_mixture}, $\delta_2^*$ is not identified when $\pi^*= 1$, $\delta_1^*$ is not identified when $\pi^*= 0$, and $\pi^*$ is not identified when $\delta_1^* = \delta_2^*$. As documented by \citet{chowhite10joe}, this lack of identification under the null leads to several undesirable properties of the LRT in the context of censored duration models. 

First, the asymptotic null distribution of the LRT statistic is unknown when censoring is covariate-dependent.\footnote{
\citet{chowhite10joe} derive the asymptotic distribution of the LRT statistic under fixed censoring and random censoring.} Second, even in the fixed and random censoring cases, the asymptotic null distribution of the LRT statistic is nonstandard and must be simulated because it is a functional of a Gaussian process whose covariance structure depends on the covariate and censoring mechanism in a complex way. Third, the LRT statistic and its asymptotic null distribution depend on the parameter space employed in maximizing the log-likelihood, so the choice of the parameter space affects the power of the LRT. Fourth, the LRT requires the lower bound of the parameter space of $\delta$ to exceed $0.5\delta^*$, a technical condition that ensures a finite variance of the likelihood ratio but has no economic justification.

Given these problems associated with the LRT, we focus on testing $\delta_{1}^* = \delta_{2}^*$ in $H_0$ in \eqref{eq:null_hypothesis}. Focusing on $\delta_{1}^* = \delta_{2}^*$ involves no loss of generality for the null data-generating process, because under either branch of $H_0$ the observed data are generated by a single Weibull density, which can always be represented with $\delta_{1}^* = \delta_{2}^*$. This restriction may reduce the test's power against alternatives with $\pi^*(1-\pi^*) \simeq 0$, such as discrete mixture 2 in Section \ref{sec:MC}. Our simulations indicate that this loss is small: under that design, the EM test and the LRT have similarly low power.

Define the log-likelihood function under the null hypothesis as
\begin{equation} \label{loglik0}
L_{0n}(\delta, \bs{\beta}, \gamma) \coloneq \sum_{t=1}^n \log f^{c}(Y_t^{c},D_t | \bs{X}_t;\delta,\bs{\beta},\gamma),
\end{equation}
where $f^{c}(y^{c},d | \bs{x};\delta,\bs{\beta},\gamma)$ is defined in \eqref{eq:model1}. Define the maximum likelihood estimator (MLE) of the parameters under the null hypothesis as
\[
(\widehat \delta, \widehat{\bs{\beta}}, \widehat \gamma) \coloneq \mathop{\arg\max}_{(\delta,\bs{\beta},\gamma) \in D \times B \times \Gamma} L_{0n}(\delta, \bs{\beta}, \gamma).
\]
Define the log-likelihood function under the two-point discrete mixture model as
\begin{equation} \label{loglik}
L_{n}(\pi, \delta_{1}, \delta_{2}, \bs{\beta}, \gamma) \coloneq \sum_{t=1}^n \log f_{a}^{c}(Y_t^{c},D_t | \bs{X}_t; \pi, \delta_{1}, \delta_{2}, \bs{\beta}, \gamma),
\end{equation}
with $f^{c}_{a}(y^{c},d | \bs{x};\pi,\delta_{1}, \delta_{2},\bs{\beta},\gamma)$ defined in \eqref{eq:model_mixture}. Define the penalized log-likelihood function under the two-point discrete mixture model as
\begin{equation} \label{pen_loglik}
PL_{n}(\pi, \delta_{1}, \delta_{2}, \bs{\beta}, \gamma) \coloneq L_{n}(\pi, \delta_{1}, \delta_{2}, \bs{\beta}, \gamma) + p(\pi),
\end{equation}
where $p(\pi)$ is a penalty function on $\pi$. We specify $p(\pi)$ in Section \ref{sec:MC}. Its purpose is to penalize values of $\pi$ that are too close to zero or one and to ensure that the model stays away from regions of the parameter space where identification is weak or lost.

Let $\bs{\theta} \coloneq (\delta_1, \delta_2, \bs{\beta}\t, \gamma)\t \in \Theta\coloneq D^2 \times B \times \Gamma$ and $\bs{\theta}^* \coloneq (\delta^*, \delta^*, (\bs{\beta}^*)\t, \gamma^*)\t$. Then, the parameter of model \eqref{eq:model_mixture} is given by $(\pi,\bs{\theta})$, and the penalized log-likelihood function is written as $PL_{n}(\pi,\bs{\theta})$. Let $\Pi$ be a finite subset of $(0,0.5]$. For each $\pi_0 \in \Pi$, define $\pi^{(1)}(\pi_0)=\pi_0$, and define $\bs{\theta}^{(1)}(\pi_0)$ as
\[
\bs{\theta}^{(1)}(\pi_0) \coloneq \mathop{\arg\max}_{\bs{\theta} \in \Theta} PL_{n}(\pi^{(1)}(\pi_0), \bs{\theta}).
\]
Starting from $(\pi^{(1)}(\pi_0),\bs{\theta}^{(1)}(\pi_0))$, we iteratively update $\pi$ and $\bs{\theta}$ using EM steps as in \citet{lcm09bm}. Suppose that $\pi^{(k)}(\pi_0)$ and $\bs{\theta}^{(k)}(\pi_0)$ have been calculated. For $t=1,\ldots,n$, define the weights for the E-step as
\begin{align*}
w_{1t}^{(k)}(\pi_0) &\coloneq \frac{\pi^{(k)}f^c(Y_t^{c}, D_t| \bs{X}_t; \delta_1^{(k)}(\pi_0), \bs{\beta}^{(k)}(\pi_0), \gamma^{(k)}(\pi_0) )}{f_{a}^c(Y_t^{c},D_t| \bs{X}_t; \pi^{(k)}(\pi_0), \bs{\theta}^{(k)}(\pi_0))}, \quad w_{2t}^{(k)}(\pi_0) \coloneq 1- w_{1t}^{(k)}(\pi_0).
\end{align*}
In the M-step, update $\pi$ and $\bs{\theta}$ as follows
\begin{align*}
\pi^{(k+1)}(\pi_0) &\coloneq \mathop{\arg \max}_{\pi \in [0,1]} \left\{ \sum_{t=1}^n w_{1t}^{(k)} \log(\pi) + \sum_{t=1}^n w_{2t}^{(k)} \log (1-\pi) + p(\pi) \right\},
\end{align*}
where we suppress the dependence of $w_{1t}^{(k)}(\pi_0)$ on $\pi_0$ for notational simplicity. The update for $\bs{\theta}$ is given by
\begin{align*}
\bs{\theta}^{(k+1)}(\pi_0) &\coloneq \mathop{\arg \max}_{\bs{\theta}\in \Theta} \left\{ \sum_{t=1}^n \sum_{j=1}^{2} w_{jt}^{(k)} \log f^c(Y_t^{c},D_t| \bs{X}_t; \delta_j, \bs{\beta}, \gamma) \right\}.
\end{align*}
Note that an EM step never decreases the penalized log-likelihood. For a pre-specified integer $K \geq 1$, define the test statistic for $\pi_0 \in \Pi$ after $K-1$ EM steps as
\[
M_n^{K}(\pi_0) \coloneq 2\left\{PL_{n}(\pi^{(K)}(\pi_0),\bs{\theta}^{(K)}(\pi_0)) - L_{0n}(\widehat \delta, \widehat{\bs{\beta}}, \widehat \gamma) \right\},
\]
where $L_{0n}(\widehat{\delta}, \widehat{\bs{\beta}}, \widehat{\gamma})$ is the maximized log-likelihood under the null. 

The \textit{EM test statistic} is defined as the maximum of $M_n^{K}(\pi_0)$ over $\pi_0\in \Pi$,
\[
\text{EM}_n^{K} \coloneq \max \left\{M_n^{K}(\pi_0): \pi_0 \in \Pi \right\}.
\]
The following proposition gives the asymptotic null distribution of the EM test statistic. Assumptions \ref{assn1}--\ref{assn3} are presented in Appendix \ref{app:assumptions} and correspond to Assumptions A1--A3 in \citet{chowhite10joe} except that we explicitly state assumptions on $C_t$. In addition, our \Cref{assn_3.6} strengthens their Assumption A3.6 to be consistent with Assumption A.5(iii) in \citet{chowhite07em}. These assumptions allow for dependent data and are applicable to autoregressive conditional duration models.
\begin{proposition} \label{EM_stat}
Suppose that Assumptions \ref{assn1}--\ref{assn3} hold and $\delta_1^*=\delta_2^*$. Suppose that $p(\pi)$ is nonpositive and continuous, $p(0.5)=0$, and $p(\pi) \to -\infty$ as $\pi \to 0$ or $\pi \to 1$. Suppose further that $0.5 \in \Pi$. Then, for any fixed finite $K$,
\[
\text{EM}_{n}^{K} \rightarrow_d (\max\{0,N(0,1)\})^2.
\]
\end{proposition}

\subsection{LRT for unobserved heterogeneity}\label{sec:LRT}

\citet{chowhite10joe} derive the asymptotic null distribution of the LRT statistic under fixed and random censoring. This section reviews their results. Let $A = [\alpha_L,\alpha_H]$, where $0 < \alpha_L \leq 1 \leq \alpha_H < \infty$, denote the set of admissible values for $\alpha_1=\delta_1/\delta^*$ and $\alpha_2=\delta_2/\delta^*$.

Define the MLE of the two-point discrete mixture model as
\[
(\widehat \pi, \widehat \delta_1, \widehat \delta_2, \widehat{\bs{\beta}}_{a}, \widehat \gamma_{a}) \coloneq \mathop{\arg\max}_{(\pi, \delta_1,\delta_2,\bs{\beta},\gamma) \in [0,1] \times \widehat{A} \times \widehat{A} \times B \times \Gamma} L_{n} (\pi, \delta_{1}, \delta_{2}, \bs{\beta}, \gamma) ,
\]
where $L_{n} (\pi, \delta_{1}, \delta_{2}, \bs{\beta}, \gamma)$ is defined in \eqref{loglik}, and $\widehat{A} \coloneq [\alpha_L \widehat\delta, \alpha_H \widehat\delta]$. The LRT statistic for testing $H_0$ is 
\[
LR_n \coloneq 2 \left[ L_{n} (\widehat \pi, \widehat \delta_1, \widehat \delta_2, \widehat{\bs{\beta}}_{a}, \widehat \gamma_{a}) - L_{0n}(\widehat \delta, \widehat{\bs{\beta}}, \widehat \gamma) \right].
\]

\citet[][p.\ 465]{chowhite10joe} show that, if $\inf A > 1/2$, the asymptotic null distribution of $LR_n$ is given by
\begin{equation} \label{eq:cho_white_g}
\sup_{\alpha \in A}\left( \max[0,\mathcal{G}^c(\alpha)] \right)^2,
\end{equation}
where $\mathcal{G}^c(\alpha)$ is a Gaussian process with covariance kernel $E[\mathcal{G}^c(\alpha)\mathcal{G}^c(\alpha')]$. As noted by \citet[][p.\ 465]{chowhite10joe}, this covariance depends on the joint distribution of $(\bs{X}_t, C_t)$ and $(\alpha,\alpha')$. Due to the difficulty of consistently estimating $E[\mathcal{G}^c(\alpha)\mathcal{G}^c(\alpha')]$ across all grid points $(\alpha,\alpha')$, \citet{chowhite10joe} omit direct analysis of the covariance kernel and instead use the weighted bootstrap procedure of \citet{hansen96em} to obtain critical values. A parametric bootstrap would require generating censored samples under the null. This is feasible under fixed censoring and random censoring. Under covariate-dependent censoring, however, it would require the conditional distribution $G(\cdot|\bs{x})$, which the model leaves unspecified and whose estimation would require a nonparametric first step that smooths over $\bs{x}$. The weighted bootstrap avoids this step by reweighting the scores computed from the observed sample rather than generating new observations.

\section{Monte Carlo experiments}\label{sec:MC}

In this section, we use Monte Carlo simulations to compare the size and power properties of the EM test, the LRT, the IM test of \citet{Chesher1984} and \citet{Lancaster1984}, and the LM tests of \citet{Sharma1987}. 

\subsection{Covariate-independent censoring}
We first consider the setting in which the censoring process is independent of the covariates. Specifically, we assume that, conditional on the covariate $\bs{X}_t$ and the heterogeneity parameter $\delta$, the uncensored duration $Y_t$ follows a Weibull distribution and the censoring variable $C_t$ follows an exponential distribution. Their respective density functions are given by
\begin{equation}\label{dgp_independent}
\begin{aligned}
	f(y | \bs{x}; \delta, \bs{\beta}, \gamma) &=
	\begin{cases} \delta \gamma \exp(\bs{x}\t\bs{\beta}) y^{\gamma - 1} \exp \left(-\delta \exp(\bs{x}\t\bs{\beta})y^{\gamma} \right) & y > 0, \\
	0 & y \leq 0,
	\end{cases} \\
	g(c|\bs{x}; \lambda) &= 
	\begin{cases} 
	\lambda \exp(-\lambda c) & c \geq 0, \\
	0 & c < 0.
	\end{cases} 
\end{aligned}
\end{equation}
We assume that $C_t$ is independent of $(Y_t, \bs{X}_t)$, and that $\delta$ is i.i.d.\ across observations and independent of $(\bs{X}_t, C_t)$. The covariate $\bs{X}_t$ is one-dimensional and drawn from the standard normal distribution. The parameters $(\bs{\beta}, \gamma)$ are unknown. The parameter $\delta$ captures unobserved heterogeneity: under the null hypothesis of homogeneity, $\delta$ is constant, whereas under the alternative, $\delta$ follows a non-degenerate distribution, so that the distribution of $Y_t$ given $\bs{X}_t$ is a mixture of Weibull distributions. We specify the distribution of $\delta$ under the alternative below. For each sample size $n$, we generate $(Y_t, C_t, \bs{X}_t)$ from this process and conduct the tests on the observed data $\{(Y^c_t, D_t, \bs{X}_t)\}_{t=1}^n$.

We describe the implementation details of the EM test, the LRT, the IM test, and the LM tests. In the EM test, we set $\Pi = \{ 0.05, 0.1, 0.15, 0.2, 0.25, 0.3, 0.35, 0.4, 0.45, 0.5 \}$ and use the penalty function $p(\pi) = 10 \log (1 - |1 - 2 \pi|)$, following the functional form suggested by \citet{lcm09bm}. We consider $K = 1, 2, 3$, corresponding to zero, one, and two EM steps.

The LRT is implemented following the procedure described in Section \ref{sec:LRT}. For the specification of the set of admissible values $A$, we adopt three options from the simulation setup in \citet{chowhite10joe}: a narrower range $[7/9, 2]$, a moderate range $[2/3, 3]$, and a wider range $[5/9, 4]$. Each of these intervals contains 1 and satisfies $\inf A >1/2$. As the asymptotic null distribution in equation \eqref{eq:cho_white_g} is intractable, we obtain critical values using the weighted bootstrap procedure of \citet{hansen96em}. First, we discretize $A$ using equally spaced grid points $0.01$ apart. For each grid point $\alpha \in A$, we compute the score $\widehat S_{nt} (\alpha) \coloneq \{ \widehat D_{nt} (\alpha) \}^{-1/2} \widehat W_{nt} (\alpha)$, where
\begin{align*} 
\widehat D_{nt} (\alpha) \coloneq \ &\frac{1}{n} \sum_{t = 1}^n (1 - \widehat R_{nt} (\alpha))^2 - \frac{1}{n} \sum_{t = 1}^n (1 - \widehat R_{nt} (\alpha)) \widehat{\bs{U}}_{nt}\t \left( \frac{1}{n} \sum_{t = 1}^n \widehat{\bs{U}}_{nt} \widehat{\bs{U}}_{nt}\t \right)^{-1} \\
	&\times \frac{1}{n} \sum_{t = 1}^n \widehat{\bs{U}}_{nt} (1 - \widehat R_{nt} (\alpha)), \\
	\widehat W_{nt} (\alpha) \coloneq \ &(1 - \widehat R_{nt} (\alpha)) - \widehat{\bs{U}}_{nt}\t \left( \frac{1}{n} \sum_{t = 1}^n \widehat{\bs{U}}_{nt} \widehat{\bs{U}}_{nt}\t \right)^{-1} \frac{1}{n} \sum_{t = 1}^n \widehat{\bs{U}}_{nt} (1 - \widehat R_{nt} (\alpha)), \\
	\widehat R_{nt} (\alpha) \coloneq \ &f^{c} (Y_t^c, D_t| \bs{X}_t; \alpha \widehat \delta, \widehat{\bs{\beta}}, \widehat \gamma) / f^{c} (Y_t^c, D_t| \bs{X}_t; \widehat \delta, \widehat{\bs{\beta}}, \widehat \gamma), \\
	\widehat{\bs{U}}_{nt} \coloneq \ &\nabla_{(\delta,\bs{\beta}, \gamma)\t} \log f^{c} (Y_t^c, D_t| \bs{X}_t; \widehat \delta, \widehat{\bs{\beta}}, \widehat \gamma), 
\end{align*}
and $f^c$ is defined in \eqref{eq:model1}. We then generate a sequence of i.i.d.\ standard normal random variables $\{ Z_{jt} : j=1,\ldots, J; t = 1, \ldots, n \}$ and simulate the asymptotic null distribution of $LR_n$ using the empirical distribution of 
\[
\mathcal{LR}_{jn}(A) \coloneq \sup_{\alpha \in A} \left( \max \left[ 0, \frac{1}{\sqrt{n}} \sum_{t = 1}^n \widehat S_{nt} (\alpha) Z_{jt} \right] \right)^2, \quad j = 1, \dots, J.
\]
Following \citet{chowhite10joe}, we set the bootstrap sample size to $J=500$ in all simulations.

For the IM test, we follow \citet{chowhite10joe} and implement a version of the IM test statistic of \citet{Lancaster1984} that focuses on the diagonal element associated with $\delta$. The IM test statistic is defined as $IM_n \coloneq \bs{\iota}\t \bs{W}(\bs{W}\t \bs{W})^{-1}\bs{W}\t \bs{\iota}$, where $\bs{\iota}$ is an $n \times 1$ vector of ones, and $\bs{W}$ is defined as
\begin{equation} 
	\bs{W} \coloneq 
	\begin{pmatrix} 
		\widehat{\bs{S}}_1 \t & \widehat H_1 \\
		\vdots & \vdots \\
		\widehat{\bs{S}}_n \t & \widehat H_n
	\end{pmatrix}, \notag
\end{equation}
with $\widehat{\bs{S}}_t \coloneq \nabla_{(\delta, \bs{\beta}, \gamma)\t} \log f^{c} (Y_t^c, D_t| \bs{X}_t; \widehat \delta, \widehat{\bs{\beta}}, \widehat \gamma)$ and
\[
\widehat H_{t} \coloneq \left[ \frac{\partial}{\partial \delta} \log f^{c} (Y_t^c, D_t| \bs{X}_t; \widehat \delta, \widehat{\bs{\beta}}, \widehat \gamma) \right]^2 + \frac{\partial^2}{\partial \delta^2} \log f^{c} (Y_t^c, D_t| \bs{X}_t; \widehat \delta, \widehat{\bs{\beta}}, \widehat \gamma).
\]
Under the null hypothesis, $IM_n$ converges in distribution to $\chi^2_1$.

For the LM tests, we follow \citet{chowhite10joe} and implement a version of the LM test statistic of \citet{Sharma1987} that accounts for parameter estimation error, as noted by \citet{Prieger00wp}. For $t = 1, \dots, n$ and $q=2,3$, define
\[ 
\widehat{\bs{L}}_{q, t} \coloneq \left[ L_2 \left(\{ \widehat \delta \exp( \bs{X}_t\t \widehat{\bs{\beta}})\}^{1/\widehat \gamma}Y_t^c; D_t, \widehat \gamma \right), \dots, L_q \left(\{ \widehat \delta \exp( \bs{X}_t\t \widehat{\bs{\beta}})\}^{1/\widehat \gamma}Y_t^c; D_t, \widehat \gamma \right) \right]\t,
\]
where the subscript $q$ denotes the order of the Laguerre polynomials, and the censoring-adjusted Laguerre polynomials are defined as
\begin{align*} 
	L_2(x; D, \gamma) &\coloneq \frac{1}{2} (x^{2 \gamma} - 4 x^{\gamma} + 2) + (1 - D) (x^{\gamma} - 1), \\
	L_3(x; D, \gamma) &\coloneq \frac{1}{6} (- x^{3 \gamma} + 9 x^{2 \gamma} - 18 x^{\gamma} + 6) + \frac{1}{2} (1 - D) (-x^{2 \gamma} + 4 x^{\gamma} - 2). 
\end{align*}
Define $\overline{\bs{L}}_{q, n} \coloneq n^{-1} \sum_{t = 1}^n \widehat{\bs{L}}_{q, t}$, $\widehat{\bs{W}}_{q, n} \coloneq \widehat{\bs{V}}_{q, n} - \widehat{\bs{H}}_{q, n} \widehat{\bs{B}}_n^{-1} \widehat{\bs{H}}_{q, n}\t$, where $\widehat{\bs{V}}_{q, n} \coloneq n^{-1} \sum_{t = 1}^n (\widehat{\bs{L}}_{q,t} - \overline{\bs{L}}_{q, n}) (\widehat{\bs{L}}_{q, t} - \overline{\bs{L}}_{q, n})\t$, $\widehat{\bs{H}}_{q, n} \coloneq n^{-1} \sum_{t = 1}^n (\widehat{\bs{L}}_{q, t} - \overline{\bs{L}}_{q, n}) \widehat{\bs{S}}_{t}\t$, and $\widehat{\bs{B}}_n \coloneq n^{-1} \sum_{t = 1}^n \widehat{\bs{S}}_t\widehat{\bs{S}}_t\t$. Define the LM test statistics as $LM_{q,n} \coloneq n \overline{\bs{L}}_{q, n}\t \{ \widehat{\bs{W}}_{q, n} \}^{-1} \overline{\bs{L}}_{q, n}$. Under the null hypothesis, $LM_{q,n}$ converges in distribution to $\chi_{q - 1}^2$. For $q = 2$, because $\widehat{\bs{L}}_{2,t} = (\widehat\delta^2/2)\widehat H_t + \widehat\delta\, \widehat S_{\delta,t}$, where $\widehat S_{\delta,t}$ is the first element of $\widehat{\bs{S}}_t$, and $\sum_{t=1}^n \widehat{\bs{S}}_t = \bs{0}$ at the null MLE, we have $LM_{2,n} = IM_n/(1 - IM_n/n)$. Consequently, the IM and LM$_2$ tests have identical size-adjusted power.
 
Table \ref{table level independent} reports the empirical rejection rates of the tests under the null hypothesis. The parameters of the data-generating process \eqref{dgp_independent} are set to $(\delta, \bs{\beta}, \gamma, \lambda) = (1,1,1,1)$, so that about half of the observations are censored, and the nominal significance level is fixed at 5\%. The EM, IM, and LM tests use asymptotic critical values, while the LRT employs bootstrapped critical values.

The EM test exhibits mild size distortion in small samples: its rejection rates are 11\% to 12\% at $n = 50$ and about 8\% at $n = 100$, regardless of $K$. The empirical size quickly approaches the nominal level as the sample size increases. For $n \geq 500$, rejection rates are close to 5\%. In contrast, the LRT is conservative across all sample sizes, particularly when the admissible set $A$ is narrow (e.g., $[7/9, 2]$). Although rejection rates increase as the range of $A$ widens, they remain slightly below the nominal level even at $n = 5000$. The IM and LM tests exhibit size distortions, particularly in small and moderate samples. For instance, the IM test rejects the null in 16.18\% of replications at $n = 50$ and 12.40\% at $n = 100$, while LM$_2$ performs similarly. LM$_3$ severely over-rejects the null even in large samples.

Table \ref{table level robust} examines how the size of the EM test reported in Table \ref{table level independent} depends on the censoring fraction and the covariate distribution. We set the censoring fraction to 25\%, 50\%, or 75\% by changing the parameter $\lambda$ in the distribution of the censoring variable $C_t$ in \eqref{dgp_independent} and draw $\bs{X}_t$ from either the standard normal distribution or the $\chi^2_1$ distribution standardized to have mean zero and unit variance. For $n \geq 500$, the rejection rates lie between 4.3\% and 5.7\% in every design, confirming that the EM test controls size well at these sample sizes. At $n = 50$, the size distortion increases with the censoring fraction.

To evaluate the power properties of the tests, we consider the following six alternative distributions for the unobserved heterogeneity parameter $\delta$ in \eqref{dgp_independent}:
\begin{enumerate}
    \item Discrete mixture 1: $\delta \sim \text{DM}(0.7370, 1.9296;\ 0.5)$
    \item Discrete mixture 2: $\delta \sim \text{DM}(0.7370, 2.4222;\ 0.95)$
    \item Gamma mixture: $\delta \sim \text{Gamma}(5,\ 5)$
    \item Log-normal mixture: $\delta \sim \text{Log-normal}(-\log(1.2)/2,\ \log(1.2))$
    \item Uniform mixture 1: $\delta \sim \text{Uniform}[0.30053,\ 2.3661]$
    \item Uniform mixture 2: $\delta \sim \text{Uniform}[1,\ 5/3]$.
\end{enumerate}
Here, $\text{DM}(a, b; p)$ denotes a two-point discrete distribution with $\Pr(\delta = a) = p$ and $\Pr(\delta = b) = 1 - p$. Discrete mixture 2 is designed to assess whether the EM test loses power when the mixing proportion is close to zero or one; the other five specifications follow \citet{chowhite10joe}. For the other parameters in \eqref{dgp_independent}, we set $(\bs{\beta}, \gamma, \lambda) = (1,1,1)$.

Tables \ref{table power independent discrete 1}--\ref{table power independent uniform 2} show the power properties of the tests. 
For the comparison of the EM test and the LRT, we focus on the size-unadjusted power for sample sizes of 500 or more, where the EM test controls the level. This comparison is the practically relevant one, because size adjustment requires knowledge of the data-generating process and is therefore infeasible in practice. In this comparison, the EM test generally performs better than the LRT, although the difference is small when the power of both tests is close to 100\%. The difference is also small under discrete mixture 2 and uniform mixture 2, where no test has appreciable power, and under discrete mixture 2 the LRT rejects slightly more often than the EM test at larger sample sizes, particularly with the wider sets $A$. The IM and LM$_2$ tests have lower power than the EM test in nearly all cases, while the high rejection rates of LM$_3$ reflect its large size distortion.

With size adjustment, the EM test and the LRT have higher power than the IM and LM tests in every design for $n \geq 500$. The size-adjusted comparison is the one more favorable to the LRT: because the EM test is already correctly sized at these sample sizes, the adjustment leaves its power essentially unchanged, whereas it raises the power of the conservative LRT substantially.
The power of the EM test for $K=2$ and $K=3$ is very similar to that for $K=1$; for brevity, Tables \ref{table power independent discrete 1}--\ref{table power independent uniform 2} report results only for $K=1$.

\subsection{Covariate-dependent censoring}\label{sec:MC_dependent}
We next examine the performance of the tests in the setting where the censoring process depends on covariates.
We specify the conditional density of $Y$ given $\bs{X}$ and $\delta$ and the censoring time $C$ as
\begin{align} 
	f(y | \bs{x}; \delta, \bs{\beta}, \gamma) &= \delta \gamma \exp(\bs{x}\t\bs{\beta})y^{\gamma - 1} \exp(- \delta \exp(\bs{x}\t\bs{\beta}) y^{\gamma}) \quad \text{for} \ y > 0, \notag \\
	C &= 2.5 \exp(-\eta_1 X_1^2 + \eta_2 X_2), \label{dgp dependent}
\end{align}
where $\bs{X} \coloneq (X_1, X_2)\t$, and $X_1$ and $X_2$ are mutually independent covariates following the chi-squared distribution with one degree of freedom and the standard normal distribution, respectively. Because $C$ is a deterministic function of the covariates, the design is covered by model \eqref{eq:model1} as the special case in which the conditional distribution of $C$ given $\bs{X}=\bs{x}$ is degenerate, so that $\nu_{\bs{x}}$ is counting measure on a single point and $g(\cdot|\bs{x})$ equals one there.
The specification of $\delta$, the observed data, and the implementation of the tests are the same as in the independent-censoring case, except that $\bs{X}_t$ is now two-dimensional.

Unlike the IM and LM tests, whose $\chi^2$ asymptotic null distributions follow from Assumptions \ref{assn1}--\ref{assn3} regardless of whether censoring depends on $\bs{X}_t$, the LRT's asymptotic null distribution has been formally established by \citet{chowhite10joe} only under fixed or random censoring.
Its validity under covariate-dependent censoring is, at present, an open question. We nonetheless compute the LRT's critical values using the same weighted bootstrap procedure as in the covariate-independent case above. The LRT results reported below therefore apply the procedure outside the setting covered by its asymptotic theory.

Table \ref{table level dependent} reports the empirical rejection rates of the tests under the null hypothesis.
We set $\delta = 1$, $\beta_1 = -0.3$, $\beta_2 = 1$, $\gamma = 1$, $\eta_1 = 0.5$, and $\eta_2 = 1$, so that about 42\% of the observations are censored.
The nominal significance level is fixed at 5\%.
The results are similar to those under independent censoring: the EM test has size close to 5\% for $n \geq 500$, the LRT is conservative, and the IM and LM tests over-reject.

To assess power, we consider the same six specifications for $\delta$ as in the independent-censoring setting. The value of $(\beta_1, \beta_2, \gamma, \eta_1, \eta_2)$ is the same as in the null setting.
Tables \ref{table power dependent discrete 1}--\ref{table power dependent uniform 2} show the power properties of the tests.
As in the independent-censoring case, the IM and LM$_2$ tests lag behind the EM test and the LRT in nearly all cases, and the higher rejection rates of LM$_3$ reflect its size distortion.
The comparison between the EM test and the LRT also follows the independent-censoring pattern: for $n \geq 500$, the EM test generally has higher size-unadjusted power than the LRT. As in the independent-censoring case, the power for $K=2$ and $K=3$ is very similar to that for $K=1$, so Tables \ref{table power dependent discrete 1}--\ref{table power dependent uniform 2} report results only for $K=1$.

\section{Real-world data analysis}\label{sec:realdata}
In this section, we compare the proposed EM test with the other tests by analyzing real-world data from the Stanford Heart Transplant Program.
The dataset is taken from \citet{kf2011book} and is also available as ``jasa'' in the R package \textbf{survival}. The dataset records the possibly censored survival times (in days) of 103 patients along with their censoring status, which takes the value one if the patient died (uncensored) and zero if the patient is censored.
The dataset also records patients' characteristics, of which we use three: the age (in years) at the time of acceptance, the year of acceptance into the program, and an indicator of prior heart surgery. These are the baseline covariates used by \citet{kf2011book} in their analysis of these data.
We do not use the patients' transplant status. A patient is eligible for a transplant from acceptance onward and joins the treatment group only when a donor heart becomes available, so transplant status is realized after acceptance. \citet{kf2011book} accordingly model it as a time-dependent covariate, which lies outside the model of Section \ref{sec:Weibull_model}.
If transplantation alters mortality, patients with the same covariates but different transplant histories have different survival distributions, which can by itself produce a rejection of the homogeneous Weibull model. 

Censoring in this dataset is administrative. Patients entered the program over a period of several years, and those still alive when follow-up ended were censored.
A patient's potential censoring time is thus the time from acceptance to the common end of follow-up, which varies with the date of acceptance.
Because 26 of the 28 censored patients share the same end-of-follow-up date, the correlation between the potential censoring time and the year of acceptance is $-0.988$.
The program also changed over this period: \citet{kf2011book} find the year of acceptance to be a significant predictor of survival in these data and attribute it to a relaxation of the admission requirements.
The tests require that $Y_t$ and $C_t$ be conditionally independent given $\bs{X}_t$. Because the potential censoring time is determined by the acceptance date, this amounts to requiring that survival be independent of the acceptance date given the covariates. We return to this requirement after introducing the covariate specifications. When the covariates include the year of acceptance, the application falls in the covariate-dependent censoring regime of Section \ref{sec:MC_dependent}, in which the EM test remains valid under Assumptions \ref{assn1}--\ref{assn3} while the asymptotic validity of the LRT has not been established; the LRT results reported below should be read with this caveat in mind.

Using this dataset, we test the specified homogeneous Weibull regression against a scale-mixture alternative.
As introduced in Section \ref{sec:Weibull_model}, the tests are based on the Weibull density function of the survival time
\begin{equation} 
	f(y | \bs{x}; \delta, \bs{\beta}, \gamma) = \delta \gamma \exp(\bs{x}\t\bs{\beta})y^{\gamma - 1} \exp(-\delta \exp(\bs{x}\t\bs{\beta})y^{\gamma}), \notag
\end{equation}
where, for the covariate vector $\bs{X}$, we consider the following five specifications
\begin{align}
	&\text{(I)} \ (X_{\mathrm{age}}), \quad
	\text{(II)} \ (X_{\mathrm{age}}, X_{\mathrm{age}}^2), \quad
	\text{(III)} \ (X_{\mathrm{age}}, X_{\mathrm{yr}}), \quad
	\text{(IV)} \ (X_{\mathrm{age}}, X_{\mathrm{age}}^2, X_{\mathrm{yr}}), \notag \\
	&\text{(V)} \ (X_{\mathrm{age}}, X_{\mathrm{age}}^2, X_{\mathrm{yr}}, X_{\mathrm{srg}}), \label{specification}
\end{align}
where $X_{\mathrm{age}}$ is the age at the time of acceptance, $X_{\mathrm{yr}}$ is the year of acceptance minus 1967, and $X_{\mathrm{srg}}$ indicates prior heart surgery. Specification (V) uses the full set of baseline covariates of \citet{kf2011book}. Only 16 of the 103 patients had prior surgery and only 9 of those 16 died, so the coefficient on $X_{\mathrm{srg}}$ is imprecisely estimated; \citet{kf2011book} report it as significant only at the 10\% level and caution that it rests on these 16 cases. For specifications (III)--(V), which include the year of acceptance, conditional independence of survival and censoring requires that the acceptance date within a year carry no further information about survival given the other covariates. We maintain this assumption. We also assume that the two patients censored before the common end of follow-up, for reasons the data do not record, were censored independently of their survival given the covariates. Specifications (I) and (II), by contrast, omit the year of acceptance. Because the year of acceptance predicts survival given age and also determines the censoring time, survival and censoring are unlikely to be conditionally independent given age alone: in a Weibull regression, the coefficient on the year of acceptance is significant at the 5\% level given age ($p = 0.04$) and at the 10\% level given age and its square ($p = 0.09$). We report specifications (I) and (II) for comparison, but they lie outside the maintained assumption, and their results should be interpreted with caution.

For the EM test, we use the same $\Pi$ and $p(\pi)$ as in Section \ref{sec:MC} and again consider $K = 1, 2, 3$. The other tests for comparison are the LRT, the IM test, and the LM tests, implemented using the same procedure as in Section \ref{sec:MC}.

Table \ref{real-world data analysis} summarizes the results of the analysis.
The upper panel reports $p$-values computed from the asymptotic null distribution of each test statistic.
Because the sample size of this dataset is $n = 103$ and Table \ref{table level dependent} shows that none of the tests has accurate size at $n = 100$, the lower panel reports $p$-values computed from the empirical distribution of each test statistic under the null setting with $n = 100$ in that simulation.
These simulation-based $p$-values are not a valid finite-sample correction for the present data. Although the asymptotic null distributions of the EM, IM, and LM test statistics do not depend on the data-generating process, their finite-sample distributions do, as the differences between Tables \ref{table level independent} and \ref{table level dependent} show. We therefore report the lower panel as a sensitivity analysis, which asks whether each test's conclusion survives the finite-sample null distribution observed in a related design.
We do not report simulation-based $p$-values for the LRT because, unlike those of the other tests, the asymptotic null distribution of the LRT statistic depends on the joint distribution of $(\bs{X}_t, C_t)$, and the critical values cannot be carried over from the simulation design in Section \ref{sec:MC_dependent} to this dataset.

In the upper panel, the EM test rejects the null hypothesis of homogeneity at the 1\% level in all five specifications.
The LRT fails to reject at the 5\% level for the narrowest admissible set $A$ in every specification.
Widening $A$ recovers rejection in specifications (I), (III), (IV), and (V), but in specification (II) all three choices of $A$ fail to reject.
In many cases, the LRT has a smaller $p$-value when $A$ is wide, which is consistent with Tables \ref{table power dependent discrete 1}--\ref{table power dependent uniform 2}, where the LRT has higher power with a wider $A$ than with a narrower one.
The EM test exhibits no comparable sensitivity to the choice of $K$.
The IM and LM$_2$ tests reject the null hypothesis in four of the five specifications, failing only in specification (II). LM$_3$ rejects in all five.

In the lower panel, the EM test continues to reject the null hypothesis at the 5\% level in all five specifications, and at the 1\% level in all specifications except (II), where its simulation-based $p$-value is $0.013$.
The IM and LM tests behave very differently.
Their simulation-based $p$-values are considerably larger than their asymptotic ones, and all three fail to reject the null hypothesis at the 5\% level except for specification (III).
This pattern is consistent with the substantial over-rejection of the IM and LM tests at $n = 100$ in Table \ref{table level dependent}, but the lower panel does not establish that their rejections in specifications (I), (IV), and (V), and that of LM$_3$ in (II), reflect size distortion for the present data. The IM and LM$_2$ tests have identical simulation-based $p$-values because $LM_{2,n} = IM_n/(1 - IM_n/n)$.

The same comparison is not available for the LRT. Table \ref{table level dependent} nonetheless indicates that the LRT is conservative at this sample size, consistent with its failure to reject the null in many cases.
Taken together, the two panels show that the EM test is the only test whose rejection of the null hypothesis of homogeneity survives both this sensitivity analysis and the choice of specification.
The EM test rejects at the 5\% level in every specification in both panels, whereas the conclusions of the IM and LM tests change between the two panels, and those of the LRT strongly depend on the choice of $A$ and of the specification.

\section{Conclusion}\label{sec:conclusion}

Unobserved heterogeneity poses a serious problem in duration analysis: ignoring it can bias the estimated structural parameters of the hazard function and invalidate inference. In censored duration models, the existing tests for unobserved heterogeneity offer an unsatisfactory remedy, as the IM and LM tests suffer from severe size distortion and the LRT requires a computationally demanding bootstrap for inference.

This paper develops an EM test for unobserved heterogeneity in censored Weibull duration models. The test builds on the EM approach of \citet{lcm09bm}, and its statistic converges in distribution to the square of $\max\{0, N(0,1)\}$ under the null hypothesis, so that critical values are obtained directly from the standard normal distribution without simulation or bootstrap. Unlike the LRT of \citet{chowhite10joe}, the EM test accommodates covariate-dependent censoring of arbitrary form without additional adjustment.

Monte Carlo simulations show that the EM test controls size well for sample sizes of 500 or more, under both covariate-independent and covariate-dependent censoring and for censoring fractions from 25\% to 75\%. Its size-adjusted power is comparable to that of the LRT and higher than that of the IM and LM tests. Size adjustment is not available in practice, however, and the LRT is markedly conservative in finite samples, so the EM test attains higher power than the LRT in most designs when the two are applied as a practitioner would apply them.

An application to survival data from the Stanford Heart Transplant Program illustrates the same pattern. The EM test rejects the null hypothesis of homogeneity in all five covariate specifications, whereas the LRT's conclusion depends on the choice of the admissible set in many cases. The rejections by the IM and LM tests, in contrast, generally do not survive when $p$-values are computed from the finite-sample null distributions in the simulations.

The test maintains the assumption that the duration and the censoring time are conditionally independent given the covariates. Under endogenous censoring, the duration distribution is not point identified without further structure. \citet{KhanTamer09joe}, \citet{Szydlowski19jae}, and \citet{Sakaguchi24jae} use instruments or report identified sets; extending the EM test in that direction is left for future work.

\section*{Funding}

This work was supported by JSPS KAKENHI Grant Number JP24K04814.

\section*{Declaration of generative AI and AI-assisted technologies in the manuscript preparation process}

During the preparation of this work the authors used Claude (Anthropic) and Codex (OpenAI) in order to proofread and copyedit the manuscript text, correct internal cross-references and notation, and check the simulation code for errors. After using these tools, the authors reviewed and edited the content as needed and take full responsibility for the content of the published article.

\bibliography{mixture_duration.bib}

\begin{appendices}

\numberwithin{equation}{section}

\section{Assumptions}\label{app:assumptions}

\begin{assumption} \label{assn1}
(i) $\{(Y_t, C_t, \bs{X}_t)\}$ is a strictly stationary geometric $\beta$-mixing process with $\beta$-mixing coefficients $\beta_\tau \leq c\rho^\tau$ for some $c>0$ and $\rho \in [0,1)$, where $Y_t$ and $C_t$ are $\mathbb{R}^+$-valued, $\bs{X}_t$ is $\mathbb{R}^k$-valued, and $\bs{X}_t$ does not contain a constant term. (ii) For $t=1,2,\ldots$, conditional on $\bs{X}_t$, $(Y_t,C_t)$ has the conditional density
\begin{align*}
m(y,c|\bs{X}_t; \pi^*,\delta_1^*,\delta_2^*,\bs{\beta}^*, \gamma^*)\coloneq & \{ \pi^* f ( y |\bs{X}_t; \delta_1^*, \bs{\beta}^*, \gamma^*) + ( 1-\pi^*) f ( y | \bs{X}_t; \delta_2^*,\bs{\beta}^*,\gamma^*) \} g(c|\bs{X}_t)
\end{align*}
for some $(\pi^*,\delta_1^*,\delta_2^*,\bs{\beta}^*,\gamma^*) \in [0,1] \times D \times D \times B \times \Gamma$, where $D \times D \times B \times \Gamma$ is a convex compact subset of $\mathbb{R}^+ \times \mathbb{R}^+ \times \mathbb{R}^d \times \mathbb{R}^+$;
\[
f (y |\bs{X}_t; \delta^*, \bs{\beta}^*, \gamma^*) = \delta^* \gamma^* \varphi(\bs{X}_t; \bs{\beta}^*)y^{\gamma^*-1}\exp(-\delta^* \varphi(\bs{X}_t;\bs{\beta}^*) y^{\gamma^*});
\]
for each $\bs{\beta} \in B$, $\varphi(\cdot;\bs{\beta}):\mathbb{R}^k \to \mathbb{R}^+$ is a Borel measurable function. Further, with $\mathcal{F}_{t-1} \coloneq \sigma(\{(Y_s, C_s, \bs{X}_s)\}_{s \leq t-1})$ denoting the history of the process, $m(\cdot,\cdot|\bs{X}_t; \pi^*,\delta_1^*,\delta_2^*,\bs{\beta}^*, \gamma^*)$ is the conditional probability density function of $(Y_t, C_t)$ given $\bs{X}_t$ and $\mathcal{F}_{t-1}$.
\end{assumption}

\begin{assumption} \label{assn2}
(i) $\varphi(\bs{X}_t; \cdot)$ is four times continuously differentiable almost surely. (ii) The maximizer
\[
(\delta^*, \bs{\beta}^*, \gamma^*) \coloneq \mathop{\arg\max}_{(\delta,\bs{\beta},\gamma) \in D \times B \times \Gamma} E[ \log f^c(Y_t^c, D_t|\bs{X}_t; \delta,\bs{\beta},\gamma)]
\]
exists, is unique, and lies in the interior of $D \times B \times \Gamma$. For each $(\pi,\bs{\theta}) \in [0, 1] \times \Theta$, $E [\ell_t (\pi, \bs{\theta})]$ exists and is finite, where
\[
\ell_t(\pi,\bs{\theta})\coloneq \log[\pi f^c(Y_t^c, D_t |\bs{X}_t;\delta_1,\bs{\beta},\gamma)+(1-\pi) f^c(Y_t^c, D_t|\bs{X}_t;\delta_2,\bs{\beta},\gamma)].
\]
(iii) For each $\pi \in (0,1)$, $E \sup_{\bs{\theta}} | \ell_t (\pi, \bs{\theta})| <\infty$. (iv) If $\delta_1^* = \delta_2^*$, then, for each $\pi \in (0,1)$, $\bs{\theta}^*$ uniquely maximizes $E [\ell_t (\pi, \bs{\theta})]$. (v) The matrix $E[\bs{s}_t\bs{s}_t\t]$ is positive definite, where
\[
\bs{s}_t \coloneq \left( \left\{ \frac{\nabla_{(\delta, \bs{\beta}\t,\gamma)\t} f^{c}_t}{f^{c}_t} \right\}\t, \ \frac{\nabla_{\delta\delta} f^{c}_t}{f^{c}_t} \right)\t, \qquad f^{c}_t \coloneq f^{c}(Y_t^c, D_t |\bs{X}_t;\delta^*,\bs{\beta}^*,\gamma^*).
\]
\end{assumption}

\begin{assumption} \label{assn3}
There exists a sequence of positive, strictly stationary, and ergodic random variables $\{M_t\}$ such that for some $\epsilon > 0$, 
\begin{assumptionitems}
\item $E[M_t^{1+\epsilon}] < \infty$;
\item $\sup_{(\pi,\bs{\theta})} \left| \nabla_{j} \ell_t (\pi, \bs{\theta}) \nabla_k \ell_t (\pi, \bs{\theta}) \right| \leq M_t$;
\item $\sup_{(\pi,\bs{\theta})} \left| \nabla_{j,k} \ell_t (\pi, \bs{\theta}) \right| \leq M_t$;
\item $\left| \nabla_{i_1} f^c(Y_t^c, D_t| \bs{X}_t;\delta^*,\bs{\beta}^*,\gamma^*)/f^c(Y_t^c, D_t| \bs{X}_t;\delta^*,\bs{\beta}^*,\gamma^*) \right|^4 \leq M_t$;
\item $\left| \nabla_{i_1}\nabla_{i_2} f^c(Y_t^c, D_t | \bs{X}_t;\delta^*,\bs{\beta}^*,\gamma^*)/f^c(Y_t^c, D_t | \bs{X}_t;\delta^*,\bs{\beta}^*,\gamma^*)\right|^2 \leq M_t$;
\item \label{assn_3.6} $\sup_{(\delta,\bs{\beta},\gamma)} \left| \nabla_{i_1}\nabla_{i_2} \nabla_{i_3} f^c(Y_t^c, D_t| \bs{X}_t;\delta,\bs{\beta},\gamma)/f^c(Y_t^c, D_t | \bs{X}_t;\delta,\bs{\beta},\gamma) \right|^2 \leq M_t$;
and
\item $\sup_{(\delta,\bs{\beta},\gamma)} \left| \nabla_{i_1}\nabla_{i_2} \nabla_{i_3} \nabla_{i_4} f^c(Y_t^c, D_t| \bs{X}_t;\delta,\bs{\beta},\gamma)/f^c(Y_t^c, D_t | \bs{X}_t;\delta,\bs{\beta},\gamma) \right| \leq M_t$,
\end{assumptionitems}
where $j, k \in \{ \pi,\delta_1,\delta_2, \beta_1,\ldots,\beta_d,\gamma\}$ and $i_1,\ldots,i_4 \in \{\delta,\beta_1,\ldots,\beta_d,\gamma\}$.
\end{assumption}

\section{Proofs}\label{app:proofs}

\begin{proof}[Proof of Proposition \ref{EM_stat}]

We first show $\bs{\theta}^{(1)}(\pi_0) \to_p \bs{\theta}^*$. For any $\pi_0 \in (0,1)$, we have \\$\sup_{\bs{\theta}} \left| n^{-1}PL_{n}(\pi_0, \bs{\theta}) - E [\ell_t (\pi_0, \bs{\theta})] \right| \to_p 0$ from Lemma A1 of \citet{chowhite07em} and $|p(\pi_0)|<\infty$ because, as discussed in the proof of Theorem 1 of \citet{chowhite10joe} (see p.\ 1 of their mathematical appendix), our Assumptions \ref{assn1}--\ref{assn3} are sufficient for Assumptions A1, A2(i, iii), A3, A4, and A5(i) of \citet{chowhite07em}. Then, under Assumptions \ref{assn1}--\ref{assn2}, $\bs{\theta}^{(1)}(\pi_0) \to_p \bs{\theta}^*$ follows from Theorem 2.1 of \citet{neweymcfadden94hdbk}, because $E [\ell_t (\pi_0, \bs{\theta})]$ is continuous in $\bs{\theta}$ by Assumption \ref{assn2}. The same argument applies when $\Theta$ is replaced by its intersection with a closed neighborhood of $\bs{\theta}^*$, because $\bs{\theta}^*$ also uniquely maximizes $E [\ell_t (\pi_0, \bs{\theta})]$ over that set, and likewise in the proof of Lemma \ref{pi_update}; the remaining steps use the maximization only in a neighborhood of $\bs{\theta}^*$.

We proceed to derive the asymptotic distribution of the EM test statistic. Observe that the derivatives of $\log f_{a}^{c}(y^{c},d | \bs{x}; \pi, \delta_1, \delta_2, \bs{\beta}, \gamma)$ with respect to $\delta_1$ and $\delta_2$ are linearly dependent when $\delta_1=\delta_2=\delta$
\begin{align*}
\left. \nabla_{\delta_1} \log f_{a}^{c}(y^{c},d | \bs{x}; \pi, \delta_1, \delta_2, \bs{\beta}, \gamma)\right|_{\delta_1=\delta_2=\delta} =& \frac{\pi\nabla_{\delta} f^{c}(y^{c},d | \bs{x};\delta,\bs{\beta},\gamma)}{f^{c}(y^{c},d | \bs{x};\delta,\bs{\beta},\gamma)}, \\
\left. \nabla_{\delta_2} \log f_{a}^{c}(y^{c},d | \bs{x}; \pi, \delta_1, \delta_2, \bs{\beta}, \gamma)\right|_{\delta_1=\delta_2=\delta} =& \frac{(1-\pi) \nabla_{\delta} f^{c}(y^{c},d | \bs{x};\delta,\bs{\beta},\gamma)}{f^{c}(y^{c},d | \bs{x};\delta,\bs{\beta},\gamma)}.
\end{align*}
As a result, the Fisher information matrix is singular, and $L_n(\pi,\delta_1,\delta_2,\bs{\beta},\gamma)$ cannot be analyzed by a standard second-order Taylor expansion.

We derive a quadratic approximation of $L_n(\pi,\delta_1,\delta_2,\bs{\beta},\gamma)$ by a quadratic form of polynomials of the parameter $\bs{\theta}$. We use a reparameterization that generalizes \citet{rotnitzky00bernoulli}, who derive the asymptotics of the LRT statistic of a scalar-parameter model (i.e., without $(\bs{\beta},\gamma)$) when the Fisher information matrix is singular. Consider the following one-to-one reparameterization
\begin{equation}
\left(
\begin{array}{c}
\lambda \\
\nu \\
\end{array}
\right)\coloneq 
\left(
\begin{array}{c}
\delta_1 - \delta_2 \\
\pi \delta_1 + (1-\pi) \delta_2\\
\end{array}
\right),
\quad \text{so that} \quad
\left(
\begin{array}{c}
\delta_1\\
\delta_2 \\
\end{array}
\right)=
\left(
\begin{array}{c}
\nu+ (1-\pi) \lambda\\
\nu- \pi\lambda\\
\end{array}
\right). \label{repara}
\end{equation}
Collect the reparameterized parameters except for $\pi$ into one vector as
\[
\bs{\psi} \coloneq (\lambda, \nu, \bs{\beta}\t, \gamma)\t\in \Theta_{\psi},
\]
where $\Theta_{\psi}\coloneq \{\bs{\psi}: \bs{\beta} \in B,\ \gamma \in \Gamma,\ \nu+(1-\pi)\lambda\in D \text{ and } \nu-\pi\lambda\in D\}$, and we suppress the dependence of $\bs{\psi}$ and $\Theta_\psi$ on $\pi$ for notational brevity.
In the reparameterized model, the null hypothesis of $\delta_1^* = \delta_2^*$ is written as $\lambda^*= 0$. We denote the true value of $\bs{\psi}$ by $\bs{\psi}^*=(0, \delta^*, (\bs{\beta}^*)\t, \gamma^*)\t$.

Let the reparameterized conditional density of $(Y_t^c, D_t)$ given $\bs{X}_t$ and its logarithm be
\begin{align}
f_{a}^{c}(Y_t^c, D_t|\bs{X}_t; \pi,\bs{\psi}) &\coloneq \pi f^c(Y_t^c, D_t |\bs{X}_t;\nu+(1-\pi) \lambda ,\bs{\beta},\gamma)+(1-\pi) f^c(Y_t^c, D_t|\bs{X}_t;\nu -\pi \lambda ,\bs{\beta},\gamma) , \nonumber \\
\ell_t( \pi,\bs{\psi}) &\coloneq \log f_{a}^{c}(Y_t^c, D_t|\bs{X}_t; \pi, \bs{\psi}). \label{loglike}
\end{align}
Collect the parameters in $\bs{\psi}$ except for $\lambda$ as
\[
\bs{\eta} \coloneq (\nu, \bs{\beta}\t, \gamma)\t,\ \text{ so that }\ \bs{\psi} = (\lambda, \bs{\eta}\t)\t. 
\]

Evaluated at the true parameter value, the first derivative of the reparameterized log-density \eqref{loglike} with respect to $\lambda$ becomes zero
\begin{equation} \label{dvarlambda}
\nabla_{\lambda} \ell_t(\pi,\bs{\psi}^*) = \frac{(1-\pi) \pi \nabla_\delta f^c(Y_t^c, D_t |\bs{X}_t;\delta^*,\bs{\beta}^*,\gamma^*) - \pi (1-\pi) \nabla_\delta f^c(Y_t^c, D_t |\bs{X}_t;\delta^*,\bs{\beta}^*,\gamma^*)}{f^c(Y_t^c, D_t |\bs{X}_t;\delta^*,\bs{\beta}^*,\gamma^*)}= 0. 
\end{equation}
On the other hand, the second derivative of $\ell_t(\pi,\bs{\psi})$ with respect to $\lambda$ under the true parameter value
\begin{equation} \label{dlambda2}
\nabla_{\lambda^2} \ell_t(\pi,\bs{\psi}^*) = \pi(1-\pi) \frac{\nabla_{\delta\delta} f^c(Y_t^c, D_t |\bs{X}_t;\delta^*,\bs{\beta}^*,\gamma^*)}{f^c(Y_t^c, D_t |\bs{X}_t;\delta^*,\bs{\beta}^*,\gamma^*)}
\end{equation}
is a mean-zero random variable and serves as the score with respect to $\lambda$. Further, the first derivative of $\ell_t(\pi,\bs{\psi})$ with respect to $\bs{\eta}$ under the true parameter value is a mean-zero non-degenerate random vector
\begin{equation} \label{dvartheta}
\nabla_{\bs{\eta}} \ell_t(\pi,\bs{\psi}^*) = \nabla_{(\delta, \bs{\beta}\t,\gamma)\t}f^c(Y_t^c, D_t |\bs{X}_t;\delta^*,\bs{\beta}^*,\gamma^*) / f^c(Y_t^c, D_t |\bs{X}_t;\delta^*,\bs{\beta}^*,\gamma^*). 
\end{equation}
Therefore, we can approximate $L_n(\pi,\delta_1,\delta_2,\bs{\beta},\gamma)$ by a function of the derivatives in \eqref{dlambda2}--\eqref{dvartheta} and $(\lambda^2, \bs{\eta})$.

Define the reparameterized log-likelihood function as $L_n(\pi, \bs{\psi}) \coloneq \sum_{t=1}^n \ell_t(\pi,\bs{\psi})$. Collect the relevant parameters and scores as
\begin{equation}\label{t-psi}
\bs{t}_{n}(\pi,\bs{\psi}) \coloneq 
\begin{pmatrix}
n^{1/2}(\bs{\eta} - \bs{\eta}^*)\\
n^{1/2}\pi(1-\pi)\lambda^2/2
\end{pmatrix}, \quad \bs{s}_t \coloneq 
\begin{pmatrix}
\nabla_{\bs{\eta}} \ell_t(\pi,\bs{\psi}^*) \\
\nabla_{\lambda^2} \ell_t(\pi,\bs{\psi}^*)/[\pi(1-\pi)]
\end{pmatrix}.
\end{equation}
Define 
\begin{equation}\label{Sn}
\bs{S}_n \coloneq n^{-1/2}\sum_{t=1}^n \bs{s}_t, \quad \bs{\mathcal{I}} \coloneq E[\bs{s}_t\bs{s}_t\t], \quad \bs{\mathcal{I}}_n \coloneq 
\begin{pmatrix}
\bs{\mathcal{I}}_{n \bs{\eta}} & \bs{\mathcal{I}}_{n\bs{\eta} \lambda}\\
\bs{\mathcal{I}}_{n\bs{\eta} \lambda}\t & \mathcal{I}_{n \lambda}
\end{pmatrix}, 
\end{equation}
where $\bs{\mathcal{I}}_{n \bs{\eta}} \coloneq -n^{-1}\nabla_{\bs{\eta} \bs{\eta}\t }L_n(\pi, \bs{\psi}^*)$, $\bs{\mathcal{I}}_{n \bs{\eta}\lambda} \coloneq -n^{-1}\nabla_{\bs{\eta} \lambda^2} L_n(\pi, \bs{\psi}^*)/ [\pi(1-\pi)]$, and \\$\mathcal{I}_{n\lambda} \coloneq -n^{-1}(1/3)\nabla_{\lambda^4 } L_n(\pi, \bs{\psi}^*)/[\pi(1-\pi)]^2$. 

We show the asymptotic distribution of $M_n^{1}(\pi_0)$. Let $\bs{\psi}^{(1)}(\pi_0)$ be the reparameterized parameter value that corresponds to $\bs{\theta}^{(1)}(\pi_0)$, and let $\bs{T}_{n} \coloneq \bs{\mathcal{I}}_{n}^{1/2}\bs{t}_{n}(\pi_0,\bs{\psi}^{(1)}(\pi_0))$. We first show $\bs{t}_{n}(\pi_0,\bs{\psi}^{(1)}(\pi_0))=O_p(1)$. The proof closely follows the proof of Theorem 1 of \citet{andrews99em}. From the definition of $\bs{\psi}^{(1)}(\pi_0)$ and Lemma \ref{P-quadratic}, we have
\begin{align*}
0 &\leq PL_n(\pi_0, \bs{\psi}^{(1)}(\pi_0)) - PL_n(\pi_0, \bs{\psi}^*)\\
 &= \bs{T}_{n}\t \bs{\mathcal{I}}_{n}^{-1/2} \bs{S}_{n} - \frac{1}{2} \|\bs{T}_{n}\|^2 + R_n(\pi_0, \bs{\psi}^{(1)}(\pi_0))\\
 &= O_p(\|\bs{T}_{n}\|) - \frac{1}{2} \|\bs{T}_{n}\|^2 + (1 + \| \bs{\mathcal{I}}_{n}^{-1/2} \bs{T}_{n}\|)^2 o_p(1)\\
 &= \|\bs{T}_{n}\|O_p(1) - \frac{1}{2} \|\bs{T}_{n}\|^2 + o_p(\|\bs{T}_{n}\|) + o_p(\|\bs{T}_{n}\|^2) + o_p(1),
\end{align*}
where the third equality holds because $\bs{\mathcal{I}}_{n}^{-1/2} \bs{S}_{n} = O_p(1)$ and, from Lemma \ref{P-quadratic} and the consistency of $\bs{\psi}^{(1)}(\pi_0)$, $R_n(\pi_0, \bs{\psi}^{(1)}(\pi_0)) = o_p((1 + \|\bs{\mathcal{I}}_{n}^{-1/2}\bs{T}_{n}\|)^2)$. Rearranging this inequality yields $\|\bs{T}_n\|^2 \leq 2 \|\bs{T}_n\| O_p(1)+o_p(1)$. Denote the $O_p(1)$ term by $\varsigma_{n}$. Then, $(\|\bs{T}_{n}\|-\varsigma_{n})^2\leq \varsigma_{n}^2 + o_p(1) = O_p(1)$; taking its square root gives $\|\bs{T}_{n}\| \leq O_p(1)$. In conjunction with $\bs{\mathcal{I}}_{n} \rightarrow_p \bs{\mathcal{I}}$, we obtain $\bs{t}_{n}(\pi_0,\bs{\psi}^{(1)}(\pi_0))= O_p(1)$.

From Lemma \ref{P-quadratic}, for any $\bs{\psi}$ with $\bs{t}_{n}(\pi_0,\bs{\psi})=O_p(1)$, we can write $2[PL_n(\pi_0, \bs{\psi}) - PL_n(\pi_0, \bs{\psi}^*)]$ as
\begin{equation} \label{LR_appn0}
2[PL_n(\pi_0, \bs{\psi}) - PL_n(\pi_0, \bs{\psi}^*)] =2\bs{t}_{n}(\pi_0, \bs{\psi})\t \bs{S}_n - \bs{t}_{n}(\pi_0, \bs{\psi})\t \bs{\mathcal{I}} \bs{t}_{n}(\pi_0, \bs{\psi})+ o_p(1).
\end{equation}
Split $\bs{t}_{n}(\pi_0,\bs{\psi})$, $\bs{S}_n$, and $\bs{\mathcal{I}}$ as
\[
\bs{t}_{n}(\pi_0,\bs{\psi}) =
\begin{pmatrix}
\bs{t}_{\bs{\eta} n}\\
t_{\lambda n}
\end{pmatrix}, \quad
\bs{S}_{n}=
\begin{pmatrix}
\bs{S}_{\bs{\eta} n}\\
S_{\lambda n}
\end{pmatrix}, \quad\bs{\mathcal{I}} =
\begin{pmatrix}
\bs{\mathcal{I}}_{\bs{\eta}} & \bs{\mathcal{I}}_{\bs{\eta} \lambda} \\
\bs{\mathcal{I}}_{\lambda \bs{\eta}} & {\mathcal{I}}_{\lambda} 
\end{pmatrix}, 
\]
and define 
\[
\bs{A}\coloneq 
\begin{pmatrix}
\bs{I} & 0 \\
- \bs{\mathcal{I}}_{\lambda \bs{\eta}}\bs{\mathcal{I}}_{\bs{\eta}}^{-1} & 1
\end{pmatrix}.
\]
Then, we can write the right-hand side of \eqref{LR_appn0} as 
\begin{align*}
& 2\bs{t}_{n}(\pi_0, \bs{\psi})\t \bs{A}^{-1}\bs{A}\bs{S}_n - \bs{t}_{n}(\pi_0, \bs{\psi})\t \bs{A}^{-1}\bs{A} \bs{\mathcal{I}} \bs{A}\t(\bs{A}^{-1})\t\bs{t}_{n}(\pi_0, \bs{\psi})+ o_p(1)\\
& = A_n(\pi_0,\bs{\psi}) + B_n(\pi_0,\bs{\psi}) + o_p(1),
\end{align*}
where
\begin{align*}
A_n(\pi_0,\bs{\psi}) &\coloneq 2\bs{t}_{\bs{\eta}.\lambda n}\t \bs{S}_{\bs{\eta}n} - \bs{t}_{\bs{\eta}.\lambda n}\t \bs{\mathcal{I}}_{\bs{\eta}} \bs{t}_{\bs{\eta}.\lambda n}, \qquad
B_n(\pi_0,\bs{\psi}) \coloneq 2 t_{\lambda n} S_{\lambda.\bs{\eta}n}- \mathcal{I}_{\lambda.\bs{\eta}}(t_{\lambda n})^2, \\
\bs{t}_{\bs{\eta}.\lambda n} &\coloneq \bs{t}_{\bs{\eta} n} + \bs{\mathcal{I}}_{\bs{\eta}}^{-1} \bs{\mathcal{I}}_{\bs{\eta} \lambda} t_{\lambda n}, \qquad
S_{\lambda.\bs{\eta}n} \coloneq S_{\lambda n} - \bs{\mathcal{I}}_{\lambda \bs{\eta}}\bs{\mathcal{I}}_{\bs{\eta}}^{-1} \bs{S}_{\bs{\eta}n}, \\
\mathcal{I}_{\lambda.\bs{\eta}} &\coloneq {\mathcal{I}}_{\lambda} - \bs{\mathcal{I}}_{\lambda \bs{\eta}}\bs{\mathcal{I}}_{\bs{\eta}}^{-1}\bs{\mathcal{I}}_{\bs{\eta} \lambda} = \text{Var}(S_{\lambda.\bs{\eta}n}).
\end{align*}
From \eqref{dvartheta}, $\bs{S}_{\bs{\eta}n}$ equals $n^{-1/2} \nabla_{(\delta, \bs{\beta}\t,\gamma)\t} L_{0n}(\delta^*,\bs{\beta}^*,\gamma^*)$. By \eqref{dlambda2} and \eqref{dvartheta}, $\bs{\mathcal{I}}$ is the matrix of Assumption \ref{assn2}(v), so $\bs{\mathcal{I}}_{\bs{\eta}}$ is invertible and $\mathcal{I}_{\lambda.\bs{\eta}}$ is strictly positive. Further, because $(\delta^*,\bs{\beta}^*,\gamma^*)$ lies in the interior of $D \times B \times \Gamma$ by Assumption \ref{assn2}(ii), the set of admissible values of $\bs{t}_{\bs{\eta}.{\lambda} n}$ approaches $\mathbb{R}^{d+2}$ and that of $t_{\lambda n}$ approaches $[0,\infty)$. Therefore, it follows from a standard argument that $A_n(\pi_0,\bs{\psi}^{(1)}(\pi_0)) = 2[L_{0n}(\widehat \delta, \widehat{\bs{\beta}}, \widehat \gamma) - L_{0n}(\delta^*,\bs{\beta}^*,\gamma^{*})]+ o_p(1)$. Maximizing $B_n(\pi_0,\bs{\psi})$ under the constraint $t_{\lambda n} \geq 0$ gives $\max_{t_{\lambda n} \geq 0}B_n(\pi_0,\bs{\psi}) = (\max\{S_{\lambda.\bs{\eta}n},0\})^2/\mathcal{I}_{\lambda.\bs{\eta}}$. Finally, noting that $PL_n(\pi_0, \bs{\psi}^*)=L_{0n}(\delta^*,\bs{\beta}^*,\gamma^*) + p(\pi_0)$, we obtain
\begin{equation} \label{LR_appn2}
M_n^{1}(\pi_0)=2[PL_n(\pi_0, \bs{\psi}^{(1)}(\pi_0)) - L_{0n}(\widehat \delta, \widehat{\bs{\beta}}, \widehat \gamma)] = (\max\{S_{\lambda.\bs{\eta}n},0\})^2/\mathcal{I}_{\lambda.\bs{\eta}} + 2 p(\pi_0) +o_p(1).
\end{equation}

We proceed to show $M_n^{K}(\pi_0)=M_n^{1}(\pi_0)+o_p(1)$; we suppress $(\pi_0)$ from $\pi^{(K)}(\pi_0)$ and $\bs{\theta}^{(K)}(\pi_0)$. Because $\bs{\theta}^{(1)}(\pi_0) - \bs{\theta}^{*} = o_p(1)$ and $\pi^{(1)}(\pi_0) - \pi_0 = 0$, it follows from Lemma \ref{pi_update} and induction that $\bs{\theta}^{(K)} - \bs{\theta}^{*} = o_p(1)$ and $\pi^{(K)} - \pi_0 = o_p(1)$ for all finite $K$; in particular, $\pi^{(K)} \in \Pi_\Delta$ with probability approaching one for $\Delta \coloneq \pi_0/2$. Because a generalized EM step never decreases the likelihood value \citep{dempster77jrssb}, we have $PL_{n}(\pi^{(K)}, \bs{\theta}^{(K)}) \geq PL_{n}(\pi_0, \bs{\theta}^{(1)}(\pi_0))$, so that $M_n^{K}(\pi_0) \geq M_n^{1}(\pi_0)$. Let $\widetilde{\bs{\theta}}$ be the maximizer of $PL_{n}(\pi^{(K)}, \bs{\theta})$ in an arbitrarily small closed neighborhood of $\bs{\theta}^{*}$, and let $\widetilde{\bs{\psi}}$ be the corresponding reparameterized value. For $a,b>0$ and $\pi\in\Pi_\Delta$, $|\partial_\pi\log\{\pi a+(1-\pi)b\}|\leq\Delta^{-1}$. Consequently, with probability approaching one,
\[
\sup_{\bs{\theta}\in\Theta}\left|n^{-1}L_n(\pi^{(K)},\bs{\theta})-n^{-1}L_n(\pi_0,\bs{\theta})\right| \leq\Delta^{-1}|\pi^{(K)}-\pi_0|=o_p(1).
\]
The uniform law at $\pi_0$ established at the beginning of this proof and Assumption \ref{assn2}(iv) therefore imply $\widetilde{\bs{\theta}}\to_p\bs{\theta}^*$. Continuity of the reparameterization then gives $\widetilde{\bs{\psi}}\to_p\bs{\psi}^*$. Because $\bs{\theta}^{(K)}$ lies in this neighborhood of $\bs{\theta}^*$ with probability approaching one, $PL_{n}(\pi^{(K)}, \widetilde{\bs{\theta}}) \geq PL_{n}(\pi^{(K)}, \bs{\theta}^{(K)})$ with probability approaching one. Because $PL_n(\pi^{(K)}, \widetilde{\bs{\psi}}) \geq PL_n(\pi^{(K)}, \bs{\psi}^*)$ and parts (a) and (c) of Lemma \ref{P-quadratic} hold uniformly in $\pi \in \Pi_\Delta$, the argument used above to show $\bs{t}_{n}(\pi_0,\bs{\psi}^{(1)}(\pi_0))= O_p(1)$ applies with $(\pi^{(K)}, \widetilde{\bs{\psi}})$ in place of $(\pi_0, \bs{\psi}^{(1)}(\pi_0))$ and gives $\bs{t}_{n}(\pi^{(K)},\widetilde{\bs{\psi}})= O_p(1)$. Consequently, Lemma \ref{P-quadratic} and the decomposition into $A_n$ and $B_n$ above give
\begin{align*}
2[PL_n(\pi^{(K)}, \widetilde{\bs{\psi}}) - PL_n(\pi^{(K)}, \bs{\psi}^*)] &\leq \sup_{\bs{t} \in \mathbb{R}^{d+2} \times [0,\infty)} \left\{ 2\bs{t}\t\bs{S}_n - \bs{t}\t\bs{\mathcal{I}}\bs{t} \right\} + o_p(1)\\
&= \bs{S}_{\bs{\eta}n}\t \bs{\mathcal{I}}_{\bs{\eta}}^{-1}\bs{S}_{\bs{\eta}n} + (\max\{S_{\lambda.\bs{\eta}n},0\})^2/\mathcal{I}_{\lambda.\bs{\eta}} + o_p(1),
\end{align*}
where the supremum does not depend on $\pi$ because neither $\bs{S}_n$ nor $\bs{\mathcal{I}}$ does. Since $PL_n(\pi,\bs{\psi}^*) = L_{0n}(\delta^*,\bs{\beta}^*,\gamma^*) + p(\pi)$, $2[L_{0n}(\widehat \delta, \widehat{\bs{\beta}}, \widehat \gamma) - L_{0n}(\delta^*,\bs{\beta}^*,\gamma^*)] = \bs{S}_{\bs{\eta}n}\t \bs{\mathcal{I}}_{\bs{\eta}}^{-1}\bs{S}_{\bs{\eta}n} + o_p(1)$, and $p(\pi^{(K)}) = p(\pi_0) + o_p(1)$ by the continuity of $p$ at $\pi_0$, it follows that, with probability approaching one,
\begin{align*}
M_n^{K}(\pi_0) &\leq 2[PL_{n}(\pi^{(K)}, \widetilde{\bs{\theta}}) - L_{0n}(\widehat \delta, \widehat{\bs{\beta}}, \widehat \gamma)]\\
&\leq (\max\{S_{\lambda.\bs{\eta}n},0\})^2/\mathcal{I}_{\lambda.\bs{\eta}} + 2 p(\pi_0) + o_p(1) = M_n^{1}(\pi_0) + o_p(1),
\end{align*}
where the last equality is \eqref{LR_appn2}. Therefore, $M_n^{K}(\pi_0)=M_n^{1}(\pi_0)+o_p(1)$ holds, and $\text{EM}_{n}^{K} \rightarrow_d (\max\{0,N(0,1)\})^2$ follows from the definition of $\text{EM}_n^{K}$ because $S_{\lambda.\bs{\eta}n}/\mathcal{I}_{\lambda.\bs{\eta}}^{1/2} \to_d N(0,1)$ and $0.5 \in \Pi$. 
\end{proof}

\section{Auxiliary results}\label{app:aux}

\begin{lemma} \label{P-quadratic}
Suppose that Assumptions \ref{assn1}--\ref{assn3} hold and $\delta_1^*=\delta_2^*$. Define $\bs{t}_{n}(\pi, \bs{\psi})$, $\bs{S}_n$, $\bs{\mathcal{I}}$, and $\bs{\mathcal{I}}_n$ as in \eqref{t-psi}--\eqref{Sn}. Then, for $\pi \in (0,1)$, we can write $L_n(\pi, \bs{\psi}) -L_n(\pi, \bs{\psi}^*)$ as the sum of a quadratic function of $\bs{t}_{n}(\pi, \bs{\psi})$ and a remainder term,
\begin{equation}\label{LR}
L_n(\pi, \bs{\psi}) -L_n(\pi, \bs{\psi}^*) = \bs{t}_{n}(\pi, \bs{\psi})\t \bs{S}_n - \frac{1}{2} \bs{t}_{n}(\pi, \bs{\psi})\t \bs{\mathcal{I}}_{n} \bs{t}_{n}(\pi, \bs{\psi})+ R_n(\pi, \bs{\psi}),
\end{equation}
where, for any $\Delta \in (0,1/2)$ and with $\Pi_\Delta \coloneq [\Delta, 1-\Delta]$, (a) for any $\xi >0$,
\[
\limsup_{n \rightarrow \infty} \Pr\left(\sup_{\pi \in \Pi_\Delta} \sup_{\bs{\psi} \in\Theta_{\psi}:\| \bs{\psi}-\bs{\psi}^* \| \leq \kappa } \frac{| R_n(\pi, \bs{\psi})|}{(1 + \| \bs{t}_n(\pi, \bs{\psi}) \|)^2} > \xi \right) \rightarrow 0 \quad \text{as } \kappa \rightarrow 0,
\]
(b) $\bs{S}_n \rightarrow_d \bs{S} \sim N(0,\bs{\mathcal{I}})$, and (c) $\sup_{\pi \in \Pi_\Delta} \|\bs{\mathcal{I}}_{n} - \bs{\mathcal{I}}\| \rightarrow_p 0$.
\end{lemma}

\begin{proof}

We first show that, for $i=1,\ldots,d+2$,
\begin{align}
& \nabla_{\bs{\eta} \lambda}L_n(\pi, \bs{\psi}^*)=0, \quad \nabla_{\lambda^3} L_n(\pi, \bs{\psi}^*)=O_p(n^{1/2}), \label{nabla_1} \\
& \nabla_{\bs{\eta}\bs{\eta}\t \lambda }L_n(\pi, \bs{\psi}^*) =O_p(n), \quad \nabla_{\bs{\eta}\bs{\eta}\t\eta_i}L_n(\pi, \bs{\psi}^*)=O_p(n), \label{nabla_2}
\end{align}
and that for a neighborhood $\mathcal{N}$ of $\bs{\psi}^*$, with $\nabla^{(k)}$ denoting the $k$th derivative with respect to $\bs{\psi}$,
\begin{align}
& \sup_{\bs{\psi}\in\Theta_{\psi} \cap \mathcal{N}} \left\| n^{-1}\nabla^{(4)} L_n(\pi, \bs{\psi})- E\nabla^{(4)} \log f_{a}^{c}(Y_t^c, D_t|\bs{X}_t; \pi,\bs{\psi}) \right\|_\infty = o_p(1), \label{nabla_3}\\
& E\nabla^{(4)} \log f_{a}^{c}(Y_t^c, D_t|\bs{X}_t; \pi,\bs{\psi}) \text{ is continuous in }\bs{\psi} \in \Theta_\psi \cap \mathcal{N}. \label{nabla_4}
\end{align}
The first part of \eqref{nabla_1} follows from simple algebra. The second part of \eqref{nabla_1} follows from Assumption \ref{assn1} and $\nabla_{\lambda^3}\ell_t(\pi,\bs{\psi}^*) = \nabla_{\lambda^3} f_{a}^{c}(Y_t^c, D_t|\bs{X}_t; \pi,\bs{\psi}^*)/f_{a}^{c}(Y_t^c, D_t|\bs{X}_t; \pi,\bs{\psi}^*)$, which holds because $\nabla_\lambda f_{a}^{c}(Y_t^c, D_t|\bs{X}_t; \pi,\bs{\psi}^*)=0$. \eqref{nabla_2} follows from Assumption \ref{assn3} and the ergodic theorem. \eqref{nabla_3} is proven similarly to Lemma A1(b) of \citet{chowhite07em} using our Assumptions \ref{assn1}--\ref{assn3} in place of their Assumptions A1, A2(i, ii), A3, and A5(i). \eqref{nabla_4} follows from Assumptions \ref{assn2}--\ref{assn3}.

Expanding $L_n(\pi, \bs{\psi})$ four times around $\bs{\psi}^*$, noting that $\nabla_\lambda L_n(\pi, \bs{\psi}^*)=0$ and $\nabla_{\bs{\eta} \lambda}L_n(\pi, \bs{\psi}^*)=0$, and collecting the terms give
\begin{equation} \label{LR0}
\begin{aligned}
L_n(\pi, \bs{\psi}) -L_n(\pi, \bs{\psi}^*) =& \nabla_{\bs{\eta}}L_n(\pi, \bs{\psi}^*)(\bs{\eta} - \bs{\eta}^*) + \frac{1}{2!}(\bs{\eta} - \bs{\eta}^*)\t \nabla_{\bs{\eta} \bs{\eta}\t}L_n(\pi, \bs{\psi}^*)(\bs{\eta} - \bs{\eta}^*) \\
& + \frac{1}{2!} \nabla_{\lambda^2} L_n(\pi, \bs{\psi}^*) \lambda^2 + \frac{3}{3!} (\bs{\eta}-\bs{\eta}^*)\t \nabla_{\bs{\eta}\lambda^2} L_n(\pi, \bs{\psi}^*) \lambda^2 \\
& + \frac{1}{4!} \nabla_{\lambda^4} L_n(\pi, \bs{\psi}^*) \lambda^4 + R_n(\pi, \bs{\psi}),
\end{aligned}
\end{equation}
where, in view of \eqref{nabla_1}--\eqref{nabla_3},
\begin{align}
\lefteqn{ R_n(\pi, \bs{\psi}) = O_p(n^{1/2}) |\lambda|^3 + O_p(n) \left( \|\bs{\eta}-\bs{\eta}^*\|^2 |\lambda| + \|\bs{\eta}-\bs{\eta}^*\|^3 \right)} \label{Rn_2} \\
& +O_p(n) \left( \|\bs{\eta}-\bs{\eta}^*\|^4 + \|\bs{\eta}-\bs{\eta}^*\|^3 |\lambda|+ \|\bs{\eta}-\bs{\eta}^*\|^2 \lambda^2 + \|\bs{\eta}-\bs{\eta}^*\| |\lambda|^3 \right) \qquad \label{Rn_3} \\
& + \frac{1}{4!} \{\nabla_{\lambda^4} L_n(\pi,\bs{\psi}^\dagger) - \nabla_{\lambda^4} L_n(\pi, \bs{\psi}^*) \} \lambda^4, \label{Rn_4}
\end{align}
with $\bs{\psi}^\dag$ being between $\bs{\psi}$ and $\bs{\psi}^*$.

The terms on the right-hand side of \eqref{LR0} can be collected as in the right-hand side of \eqref{LR}. We proceed to derive the order of $R_n(\pi, \bs{\psi})$. Because $\|\bs{t}_n(\pi, \bs{\psi})\|^2 = n\|\bs{\eta} - \bs{\eta}^*\|^2 + n \pi^2(1-\pi)^2 \lambda^4/4$, the right-hand side of \eqref{Rn_2} and the terms in \eqref{Rn_3} are bounded by $O_p(1)(\|\bs{t}_n(\pi, \bs{\psi})\|+\|\bs{t}_n(\pi, \bs{\psi})\|^2)(\|\bs{\eta}-\bs{\eta}^*\|+|\lambda|)$. In view of \eqref{nabla_3} and \eqref{nabla_4}, \eqref{Rn_4} is bounded by $\|\bs{t}_n(\pi, \bs{\psi})\|^2[d(\bs{\psi}^\dagger)+o_p(1)]$ with $d(\bs{\psi}^\dagger)\rightarrow 0$ as $\bs{\psi}^\dagger \rightarrow \bs{\psi}^*$, where $d(\bs{\psi}^\dagger)$ corresponds to $n^{-1}E[\nabla_{\lambda^4}L_n(\pi, \bs{\psi}^\dagger) - \nabla_{\lambda^4}L_n(\pi, \bs{\psi}^*)]$. Therefore, $R_n(\pi, \bs{\psi})=(1+\|\bs{t}_n(\pi, \bs{\psi})\|)^2[d(\bs{\psi}^\dagger)+o_p(1) + O_p(\|\bs{\psi}-\bs{\psi}^*\|)]$, where the $o_p(1)$ and $O_p(\cdot)$ terms are uniform in $\pi \in \Pi_\Delta$ because \eqref{nabla_1}--\eqref{nabla_4} hold uniformly in $\pi \in \Pi_\Delta$. To see this, note from \eqref{repara} that $\pi$ enters $f_a^c$ only through the weights $\pi$ and $1-\pi$ and the shifts $(1-\pi)\lambda$ and $-\pi\lambda$. Hence each derivative of $L_n(\pi,\bs{\psi}^*)$ in \eqref{nabla_1}--\eqref{nabla_2} is a finite sum of polynomials in $\pi$ multiplied by sample averages that do not depend on $\pi$, and \eqref{nabla_3}--\eqref{nabla_4} extend to $(\pi,\bs{\psi}) \in \Pi_\Delta \times (\Theta_\psi \cap \mathcal{N})$ because the dominance conditions of Assumption \ref{assn3} are uniform in $\pi$. Part (a) follows.

For part (b), $\bs{s}_t$ is measurable with respect to $\mathcal{F}_t$, defined in Assumption \ref{assn1}(ii). By Assumption \ref{assn1}(ii), the conditional density of $(Y_t,C_t)$ given $\bs{X}_t$ and $\mathcal{F}_{t-1}$ equals $m(\cdot,\cdot|\bs{X}_t; \pi^*,\delta_1^*,\delta_2^*,\bs{\beta}^*,\gamma^*)$, and hence $E[\bs{s}_t|\bs{X}_t, \mathcal{F}_{t-1}] = E[\bs{s}_t|\bs{X}_t]$. Because $\nabla_\lambda f_{a}^{c}(Y_t^c, D_t|\bs{X}_t; \pi,\bs{\psi}^*)=0$, we have $\nabla_{\lambda^2}\ell_t(\pi,\bs{\psi}^*) = \nabla_{\lambda^2} f_{a}^{c}(Y_t^c, D_t|\bs{X}_t; \pi,\bs{\psi}^*)/f_{a}^{c}(Y_t^c, D_t|\bs{X}_t; \pi,\bs{\psi}^*)$, and the dominance conditions of Assumption \ref{assn3} imply $E[\nabla_{\lambda^2}\ell_t(\pi,\bs{\psi}^*)|\bs{X}_t] =0$, and $E[\nabla_{\bs{\eta}}\ell_t(\pi,\bs{\psi}^*)|\bs{X}_t]=0$ follows by the same argument. Therefore $E[\bs{s}_t|\bs{X}_t, \mathcal{F}_{t-1}]=0$, and the law of iterated expectations gives $E[\bs{s}_t|\mathcal{F}_{t-1}]=0$, so that $\{\bs{s}_t,\mathcal{F}_t\}$ is a strictly stationary and ergodic martingale difference sequence, ergodicity following from the $\beta$-mixing condition in Assumption \ref{assn1}(i). Assumption \ref{assn3} gives $E[\|\bs{s}_t\|^2]<\infty$. Part (b) then follows from the Cram\'er--Wold device and the martingale central limit theorem \citep[][Theorem 25.3]{davidson21book}.

For part (c), $\bs{\mathcal{I}}_{n \bs{\eta}} \rightarrow_p \bs{\mathcal{I}}_{\bs{\eta}}$ holds trivially. For $\bs{\mathcal{I}}_{n \bs{\eta}\lambda}$ and $\mathcal{I}_{n\lambda}$, Fa\`a di Bruno's formula and $\nabla_{\lambda} f_{a}^{c}(Y_t^c, D_t|\bs{X}_t; \pi,\bs{\psi}^*)=0$ give
\begin{align*}
\nabla_{\bs{\eta}\lambda^2}\ell_t(\pi,\bs{\psi}^*) =& \frac{\nabla_{\bs{\eta}\lambda^2}f_{a}^{c}(Y_t^c, D_t|\bs{X}_t; \pi,\bs{\psi}^*)}{f_{a}^{c}(Y_t^c, D_t|\bs{X}_t; \pi,\bs{\psi}^*)} - \frac{\nabla_{\bs{\eta}} f_{a}^{c}(Y_t^c, D_t|\bs{X}_t; \pi,\bs{\psi}^*)}{f_{a}^{c}(Y_t^c, D_t|\bs{X}_t; \pi,\bs{\psi}^*)}\frac{\nabla_{\lambda^2} f_{a}^{c}(Y_t^c, D_t|\bs{X}_t; \pi,\bs{\psi}^*)}{f_{a}^{c}(Y_t^c, D_t|\bs{X}_t; \pi,\bs{\psi}^*)} ,\\
\nabla_{\lambda^4}\ell_t(\pi,\bs{\psi}^*) =& \frac{\nabla_{\lambda^4} f_{a}^{c}(Y_t^c, D_t|\bs{X}_t; \pi,\bs{\psi}^*)}{f_{a}^{c}(Y_t^c, D_t|\bs{X}_t; \pi,\bs{\psi}^*)} - 3 \left(\frac{\nabla_{\lambda^2} f_{a}^{c}(Y_t^c, D_t|\bs{X}_t; \pi,\bs{\psi}^*)}{f_{a}^{c}(Y_t^c, D_t|\bs{X}_t; \pi,\bs{\psi}^*)}\right)^2 .
\end{align*}
Therefore, $\bs{\mathcal{I}}_{n \bs{\eta}\lambda} \to_p E[\nabla_{\bs{\eta}} \ell_t(\pi,\bs{\psi}^*)\nabla_{\lambda^2} \ell_t(\pi,\bs{\psi}^*)]/[\pi(1-\pi)]$ and $\mathcal{I}_{n \lambda} \rightarrow_p E[\{\nabla_{\lambda^2} \ell_t(\pi,\bs{\psi}^*)\}^2]/[\pi(1-\pi)]^2$, and $\bs{\mathcal{I}}_n \rightarrow_p \bs{\mathcal{I}}$ follows. The convergence is uniform in $\pi \in \Pi_\Delta$ by the argument in the proof of part (a): each element of $\bs{\mathcal{I}}_n$ is a finite sum of continuous functions of $\pi$, bounded on $\Pi_\Delta$, multiplied by sample averages that do not depend on $\pi$. 
\end{proof}

\begin{lemma} \label{pi_update}
Suppose that the assumptions of Proposition \ref{EM_stat} hold. If $\pi^{(k)}(\pi_0) - \pi_0 = o_p(1)$ and $\bs{\theta}^{(k)}(\pi_0) - \bs{\theta}^{*} = o_p(1)$, then $\pi^{(k+1)}(\pi_0) - \pi_0 = o_p(1)$ and $\bs{\theta}^{(k+1)}(\pi_0) - \bs{\theta}^{*} = o_p(1)$.
\end{lemma}
\begin{proof}
We suppress $(\pi_0)$ from $\pi^{(k)}(\pi_0)$, $\bs{\theta}^{(k)}(\pi_0)$, and so on. The proof is similar to the proof of Lemma 3 of \citet{lichen10jasa}. Let $f_t(\delta,\bs{\beta},\gamma)$ and $f_t(\pi,\bs{\theta})$ denote $f^{c}(Y_t^{c},D_t | \bs{X}_t;\delta,\bs{\beta},\gamma)$ and $f_{a}^{c}(Y_t^{c},D_t | \bs{X}_t; \pi, \delta_{1}, \delta_{2}, \bs{\beta}, \gamma)$, respectively. Recall that $\pi^{(k+1)}$ maximizes
\[
Q_n(\pi) \coloneq \sum_{t=1}^n w_{1t}^{(k)} \log(\pi) + \sum_{t=1}^n \left(1-w_{1t}^{(k)}\right) \log (1-\pi) + p(\pi).
\]
It follows from a Taylor expansion, $\bs{\theta}^{(k)} - \bs{\theta}^{*} = o_p(1)$, $\pi^{(k)} - \pi_0 = o_p(1)$, and the law of large numbers that
\begin{equation}\label{pi_up1}
\frac{1}{n} \sum_{t=1}^n w_{1t}^{(k)} = \frac{1}{n} \sum_{t = 1}^n \frac{\pi^{(k)} f_t(\delta_1^{(k)},\bs{\beta}^{(k)},\gamma^{(k)})}{f_t(\pi^{(k)}, \bs{\theta}^{(k)})}
 = \frac{1}{n} \sum_{t = 1}^n \frac{\pi_0 f_t(\delta^*,\bs{\beta}^{*},\gamma^{*})}{f_t(\pi_0, \bs{\theta}^{*})} + o_p(1)
= \pi_0 + o_p(1).
\end{equation}
Consequently, we have $Q_n(\pi) = n(\pi_0 + o_p(1)) \log(\pi) + n(1-\pi_0 + o_p(1)) \log (1-\pi) + p(\pi)$. Because $p(\pi) \to -\infty$ as $\pi \to 0, 1$ and $\log(\pi)+\log (1-\pi) \to -\infty $ as $\pi \to 0,1$, there exists $\Delta \in (0,1)$ such that $\pi^{(k+1)} \in \Pi_\Delta\coloneq [\Delta, 1-\Delta]$ with probability approaching one as $n \to \infty$. Further, $n^{-1}Q_n(\pi) \to_p Q(\pi)\coloneq \pi_0 \log (\pi) + (1-\pi_0) \log (1-\pi)$ uniformly in $\pi \in \Pi_\Delta$. Therefore, $\pi^{(k+1)} - \pi_0 = o_p(1)$ follows from Theorem 2.1 of \citet{neweymcfadden94hdbk}.

We proceed to show $\bs{\theta}^{(k+1)} - \bs{\theta}^{*} = o_p(1)$. Note that $\bs{\theta}^{(k+1)}$ maximizes \\$Q_n(\bs{\theta})\coloneq n^{-1}\sum_{t=1}^n \sum_{j=1}^{2} w_{jt}^{(k)} \log f^c(Y_t^{c},D_t| \bs{X}_t; \delta_j, \bs{\beta}, \gamma)$. By an argument similar to that for \eqref{pi_up1}, we have $\sup_{\bs{\theta} \in \Theta}|Q_n(\bs{\theta})-Q(\bs{\theta})|=o_p(1)$, where $ Q(\bs{\theta})\coloneq E[\pi_0 \log f_t(\delta_1,\bs{\beta},\gamma) + (1-\pi_0) \log f_t(\delta_2,\bs{\beta},\gamma)]$. Then, $\bs{\theta}^{(k+1)} - \bs{\theta}^{*} = o_p(1)$ follows from Theorem 2.1 of \citet{neweymcfadden94hdbk}.
\end{proof}

\end{appendices}

\singlespacing

\clearpage
\section*{Tables}

\begin{table}[H]
\centering
\caption{Covariate-independent censoring: empirical rejection rates of tests under the null}
\label{table level independent}
\small
\begin{tabular}{lrrrrrr}
\toprule
$n$ & 50 & 100 & 500 & 1000 & 2000 & 5000 \\
\midrule
EM ($K=1$) & 11.04 & 7.54 & 5.14 & 5.62 & 5.04 & 4.52 \\
EM ($K=2$) & 11.74 & 7.66 & 5.14 & 5.62 & 5.04 & 4.52 \\
EM ($K=3$) & 11.92 & 7.68 & 5.14 & 5.62 & 5.04 & 4.52 \\
LRT ($A = [7/9, 2]$) & 0.00 & 0.00 & 0.42 & 1.86 & 3.02 & 3.66 \\
LRT ($A = [2/3, 3]$) & 0.00 & 0.16 & 2.22 & 3.90 & 4.10 & 4.24 \\
LRT ($A = [5/9, 4]$) & 0.36 & 0.98 & 2.86 & 4.50 & 4.48 & 4.44 \\
IM & 16.18 & 12.40 & 8.34 & 8.02 & 6.50 & 5.90 \\
LM$_2$ & 17.92 & 13.04 & 8.52 & 8.08 & 6.52 & 5.92 \\
LM$_3$ & 34.72 & 29.66 & 21.02 & 19.10 & 16.18 & 12.88 \\
\bottomrule
\end{tabular}
\begin{flushleft}
\small
Notes: Each entry reports the empirical rejection rate of the corresponding test under the null hypothesis. The nominal level is 5\%. The weighted bootstrap sample size for the LRT is 500. The number of Monte Carlo replications is 5{,}000.
\end{flushleft}
\end{table}

\begin{table}[H]
\centering
\caption{Covariate-independent censoring: empirical rejection rates of the EM test under the null for different censoring fractions and covariate distributions}
\label{table level robust}
\small
\begin{tabular}{llrrrrrr}
\toprule
Covariate & Censored & $\lambda$ & \multicolumn{5}{c}{$n$} \\
\cmidrule(l){4-8}
 & & & 50 & 100 & 500 & 1000 & 2000 \\
\midrule
$N(0,1)$ & 25\% & 0.268 & 10.70 & 7.78 & 4.90 & 4.90 & 4.64 \\
$N(0,1)$ & 50\% & 1 & 11.04 & 7.54 & 5.14 & 5.62 & 5.04 \\
$N(0,1)$ & 75\% & 3.725 & 12.20 & 7.74 & 4.90 & 5.60 & 5.04 \\
Standardized $\chi^2_1$ & 25\% & 0.274 & 9.26 & 7.86 & 5.74 & 5.00 & 4.32 \\
Standardized $\chi^2_1$ & 50\% & 0.912 & 10.32 & 7.00 & 5.18 & 4.90 & 4.92 \\
Standardized $\chi^2_1$ & 75\% & 3.314 & 11.90 & 7.66 & 4.72 & 4.72 & 4.58 \\
\bottomrule
\end{tabular}
\begin{flushleft}
\small
Notes: Each entry reports the empirical rejection rate of the EM test with $K=1$ under the null hypothesis. The data-generating process is \eqref{dgp_independent} with $(\delta, \bs{\beta}, \gamma) = (1,1,1)$, and $\lambda$ is chosen to attain the stated censoring fraction. The row with a $N(0,1)$ covariate and 50\% censoring reproduces Table \ref{table level independent}. The rejection rates for $K=2$ and $K=3$ are identical to those for $K=1$ for $n \geq 500$ and exceed them by at most 1.3 percentage points for $n \leq 100$. The nominal level is 5\%, and the number of Monte Carlo replications is 5{,}000. With a standardized $\chi^2_1$ covariate, 75\% censoring, and $n = 50$, the test statistics cannot be computed because one replication has only two uncensored observations. The rejection rate is computed over the remaining 4{,}999 replications.
\end{flushleft}
\end{table}

\begin{table}[H]
\centering
\caption{Covariate-independent censoring: power of tests for (I) discrete mixture 1}
\label{table power independent discrete 1}
\small
\begin{tabular}{lrrrrrr}
\toprule
$n$ & 50 & 100 & 500 & 1000 & 2000 & 5000 \\
\midrule
\multicolumn{7}{l}{Size-unadjusted power} \\
\midrule
EM ($K=1$) & 14.70 & 15.42 & 33.96 & 52.46 & 78.00 & 98.34 \\
LRT ($A = [7/9, 2]$) & 0.00 & 0.04 & 9.14 & 34.00 & 69.24 & 97.94 \\
LRT ($A = [2/3, 3]$) & 0.06 & 1.26 & 20.32 & 41.80 & 70.50 & 97.60 \\
LRT ($A = [5/9, 4]$) & 0.64 & 3.72 & 23.00 & 41.84 & 68.34 & 97.28 \\
IM & 14.18 & 11.62 & 17.36 & 31.40 & 58.40 & 94.86 \\
LM$_2$ & 15.56 & 12.46 & 17.50 & 31.64 & 58.46 & 94.86 \\
LM$_3$ & 32.58 & 26.84 & 24.56 & 34.74 & 56.92 & 93.68 \\
\midrule
\multicolumn{7}{l}{Size-adjusted power} \\
\midrule
EM ($K=1$) & 6.24 & 10.74 & 33.22 & 50.42 & 77.94 & 98.52 \\
LRT ($A = [7/9, 2]$) & 10.84 & 14.62 & 35.64 & 51.06 & 77.66 & 98.56 \\
LRT ($A = [2/3, 3]$) & 9.50 & 12.56 & 31.74 & 47.44 & 74.94 & 98.02 \\
LRT ($A = [5/9, 4]$) & 9.14 & 11.30 & 30.20 & 44.72 & 70.52 & 97.42 \\
IM & 4.14 & 3.52 & 10.66 & 22.02 & 53.74 & 93.92 \\
LM$_2$ & 4.14 & 3.52 & 10.66 & 22.02 & 53.74 & 93.92 \\
LM$_3$ & 5.32 & 4.12 & 3.28 & 5.56 & 16.98 & 80.40 \\
\bottomrule
\end{tabular}
\begin{flushleft}
\small
Notes: The nominal level is 5\%, and the number of Monte Carlo replications is 5{,}000. The panel ``Size-adjusted power'' uses the 95th percentile of each test statistic's empirical null distribution in Table \ref{table level independent}. The panel ``Size-unadjusted power'' uses asymptotic critical values, bootstrapped for the LRT. Results for the EM test are reported for $K=1$ only, as the power for $K=2$ and $K=3$ is very similar.
\end{flushleft}
\end{table}

\begin{table}[h]
\centering
\caption{Covariate-independent censoring: power of tests for (II) discrete mixture 2}
\label{table power independent discrete 2}
\small
\begin{tabular}{lrrrrrr}
\toprule
$n$ & 50 & 100 & 500 & 1000 & 2000 & 5000 \\
\midrule
\multicolumn{7}{l}{Size-unadjusted power} \\
\midrule
EM ($K=1$) & 12.02 & 8.56 & 6.74 & 7.26 & 7.74 & 8.44 \\
LRT ($A = [7/9, 2]$) & 0.00 & 0.00 & 0.36 & 1.88 & 5.04 & 8.24 \\
LRT ($A = [2/3, 3]$) & 0.00 & 0.16 & 2.80 & 5.52 & 7.22 & 10.28 \\
LRT ($A = [5/9, 4]$) & 0.06 & 0.82 & 4.38 & 6.50 & 7.94 & 10.68 \\
IM & 15.10 & 11.36 & 6.82 & 6.64 & 5.78 & 5.26 \\
LM$_2$ & 16.80 & 12.02 & 7.02 & 6.74 & 5.80 & 5.26 \\
LM$_3$ & 31.08 & 26.56 & 20.10 & 19.44 & 17.22 & 16.80 \\
\midrule
\multicolumn{7}{l}{Size-adjusted power} \\
\midrule
EM ($K=1$) & 5.72 & 5.64 & 6.48 & 6.46 & 7.68 & 9.34 \\
LRT ($A = [7/9, 2]$) & 3.94 & 3.76 & 5.50 & 5.88 & 7.96 & 10.24 \\
LRT ($A = [2/3, 3]$) & 3.66 & 3.74 & 6.06 & 7.00 & 8.80 & 10.98 \\
LRT ($A = [5/9, 4]$) & 4.30 & 4.04 & 6.68 & 7.04 & 8.36 & 11.12 \\
IM & 4.36 & 3.78 & 4.10 & 3.72 & 4.40 & 4.48 \\
LM$_2$ & 4.36 & 3.78 & 4.10 & 3.72 & 4.40 & 4.48 \\
LM$_3$ & 3.54 & 3.06 & 3.94 & 3.94 & 5.08 & 7.18 \\
\bottomrule
\end{tabular}
\begin{flushleft}
\small
Notes: The nominal level is 5\%, and the number of Monte Carlo replications is 5{,}000. The panel ``Size-adjusted power'' uses the 95th percentile of each test statistic's empirical null distribution in Table \ref{table level independent}. The panel ``Size-unadjusted power'' uses asymptotic critical values, bootstrapped for the LRT.
\end{flushleft}
\end{table}

\begin{table}[h]
\centering
\caption{Covariate-independent censoring: power of tests for (III) gamma mixture}
\label{table power independent gamma}
\small
\begin{tabular}{lrrrrrr}
\toprule
$n$ & 50 & 100 & 500 & 1000 & 2000 & 5000 \\
\midrule
\multicolumn{7}{l}{Size-unadjusted power} \\
\midrule
EM ($K=1$) & 14.42 & 15.26 & 28.34 & 45.36 & 69.60 & 96.30 \\
LRT ($A = [7/9, 2]$) & 0.00 & 0.00 & 4.32 & 20.60 & 53.20 & 93.38 \\
LRT ($A = [2/3, 3]$) & 0.00 & 0.40 & 13.96 & 31.92 & 58.36 & 93.52 \\
LRT ($A = [5/9, 4]$) & 0.36 & 2.32 & 17.96 & 34.18 & 58.52 & 92.98 \\
IM & 14.94 & 12.12 & 14.76 & 25.96 & 48.46 & 89.56 \\
LM$_2$ & 16.82 & 12.82 & 14.96 & 26.06 & 48.60 & 89.56 \\
LM$_3$ & 31.40 & 25.80 & 19.42 & 23.98 & 40.50 & 83.24 \\
\midrule
\multicolumn{7}{l}{Size-adjusted power} \\
\midrule
EM ($K=1$) & 6.56 & 10.34 & 27.76 & 42.70 & 69.54 & 96.70 \\
LRT ($A = [7/9, 2]$) & 7.48 & 10.76 & 26.34 & 39.30 & 64.62 & 95.38 \\
LRT ($A = [2/3, 3]$) & 6.50 & 9.66 & 24.22 & 37.40 & 64.16 & 94.66 \\
LRT ($A = [5/9, 4]$) & 6.54 & 9.20 & 24.46 & 36.44 & 61.24 & 93.52 \\
IM & 4.98 & 4.28 & 8.46 & 17.10 & 42.74 & 88.00 \\
LM$_2$ & 4.98 & 4.28 & 8.46 & 17.10 & 42.74 & 88.00 \\
LM$_3$ & 4.54 & 3.38 & 1.90 & 3.00 & 9.02 & 60.38 \\
\bottomrule
\end{tabular}
\begin{flushleft}
\small
Notes: The nominal level is 5\%, and the number of Monte Carlo replications is 5{,}000. The panel ``Size-adjusted power'' uses the 95th percentile of each test statistic's empirical null distribution in Table \ref{table level independent}. The panel ``Size-unadjusted power'' uses asymptotic critical values, bootstrapped for the LRT.
\end{flushleft}
\end{table}

\begin{table}[h]
\centering
\caption{Covariate-independent censoring: power of tests for (IV) log-normal mixture}
\label{table power independent lognormal}
\small
\begin{tabular}{lrrrrrr}
\toprule
$n$ & 50 & 100 & 500 & 1000 & 2000 & 5000 \\
\midrule
\multicolumn{7}{l}{Size-unadjusted power} \\
\midrule
EM ($K=1$) & 13.96 & 13.00 & 22.72 & 33.72 & 52.00 & 86.06 \\
LRT ($A = [7/9, 2]$) & 0.00 & 0.02 & 3.44 & 13.88 & 38.06 & 81.38 \\
LRT ($A = [2/3, 3]$) & 0.04 & 0.42 & 11.68 & 23.30 & 42.40 & 80.92 \\
LRT ($A = [5/9, 4]$) & 0.40 & 2.04 & 14.94 & 24.00 & 41.60 & 79.18 \\
IM & 14.16 & 11.28 & 11.72 & 16.92 & 31.70 & 71.94 \\
LM$_2$ & 16.06 & 12.10 & 11.88 & 17.04 & 31.86 & 71.96 \\
LM$_3$ & 30.42 & 24.68 & 18.52 & 19.06 & 28.04 & 63.10 \\
\midrule
\multicolumn{7}{l}{Size-adjusted power} \\
\midrule
EM ($K=1$) & 6.66 & 8.62 & 22.18 & 31.72 & 51.92 & 87.20 \\
LRT ($A = [7/9, 2]$) & 7.52 & 9.32 & 21.78 & 28.76 & 48.78 & 85.58 \\
LRT ($A = [2/3, 3]$) & 6.40 & 8.28 & 20.58 & 27.20 & 47.44 & 83.36 \\
LRT ($A = [5/9, 4]$) & 6.16 & 7.90 & 20.48 & 26.12 & 44.00 & 80.56 \\
IM & 4.58 & 4.32 & 6.56 & 10.66 & 26.72 & 68.62 \\
LM$_2$ & 4.58 & 4.32 & 6.56 & 10.66 & 26.72 & 68.62 \\
LM$_3$ & 4.48 & 3.66 & 2.28 & 2.10 & 4.56 & 36.14 \\
\bottomrule
\end{tabular}
\begin{flushleft}
\small
Notes: The nominal level is 5\%, and the number of Monte Carlo replications is 5{,}000. The panel ``Size-adjusted power'' uses the 95th percentile of each test statistic's empirical null distribution in Table \ref{table level independent}. The panel ``Size-unadjusted power'' uses asymptotic critical values, bootstrapped for the LRT.
\end{flushleft}
\end{table}

\begin{table}[h]
\centering
\caption{Covariate-independent censoring: power of tests for (V) uniform mixture 1}
\label{table power independent uniform 1}
\small
\begin{tabular}{lrrrrrr}
\toprule
$n$ & 50 & 100 & 500 & 1000 & 2000 & 5000 \\
\midrule
\multicolumn{7}{l}{Size-unadjusted power} \\
\midrule
EM ($K=1$) & 17.96 & 18.82 & 49.48 & 74.12 & 94.98 & 99.98 \\
LRT ($A = [7/9, 2]$) & 0.00 & 0.00 & 16.08 & 51.32 & 88.70 & 99.90 \\
LRT ($A = [2/3, 3]$) & 0.10 & 1.30 & 30.02 & 60.74 & 90.32 & 99.90 \\
LRT ($A = [5/9, 4]$) & 0.96 & 4.14 & 35.04 & 62.72 & 89.92 & 99.90 \\
IM & 17.10 & 13.54 & 31.26 & 54.60 & 85.80 & 99.86 \\
LM$_2$ & 18.60 & 14.34 & 31.56 & 54.84 & 85.84 & 99.86 \\
LM$_3$ & 34.32 & 27.20 & 29.04 & 47.86 & 79.12 & 99.54 \\
\midrule
\multicolumn{7}{l}{Size-adjusted power} \\
\midrule
EM ($K=1$) & 8.14 & 12.78 & 48.74 & 71.90 & 94.98 & 99.98 \\
LRT ($A = [7/9, 2]$) & 13.88 & 18.90 & 50.36 & 71.00 & 93.50 & 99.98 \\
LRT ($A = [2/3, 3]$) & 12.02 & 15.48 & 44.72 & 66.84 & 93.04 & 99.96 \\
LRT ($A = [5/9, 4]$) & 11.08 & 13.64 & 44.20 & 65.34 & 91.32 & 99.90 \\
IM & 5.30 & 4.32 & 20.58 & 43.14 & 82.40 & 99.80 \\
LM$_2$ & 5.30 & 4.32 & 20.58 & 43.14 & 82.40 & 99.80 \\
LM$_3$ & 5.14 & 4.78 & 3.80 & 9.68 & 36.86 & 97.64 \\
\bottomrule
\end{tabular}
\begin{flushleft}
\small
Notes: The nominal level is 5\%, and the number of Monte Carlo replications is 5{,}000. The panel ``Size-adjusted power'' uses the 95th percentile of each test statistic's empirical null distribution in Table \ref{table level independent}. The panel ``Size-unadjusted power'' uses asymptotic critical values, bootstrapped for the LRT.
\end{flushleft}
\end{table}

\begin{table}[h]
\centering
\caption{Covariate-independent censoring: power of tests for (VI) uniform mixture 2}
\label{table power independent uniform 2}
\small
\begin{tabular}{lrrrrrr}
\toprule
$n$ & 50 & 100 & 500 & 1000 & 2000 & 5000 \\
\midrule
\multicolumn{7}{l}{Size-unadjusted power} \\
\midrule
EM ($K=1$) & 11.40 & 7.98 & 7.08 & 8.28 & 9.42 & 12.00 \\
LRT ($A = [7/9, 2]$) & 0.00 & 0.02 & 1.18 & 3.46 & 6.10 & 9.20 \\
LRT ($A = [2/3, 3]$) & 0.12 & 0.30 & 3.48 & 5.80 & 7.08 & 9.68 \\
LRT ($A = [5/9, 4]$) & 0.60 & 1.26 & 4.38 & 5.82 & 7.20 & 9.20 \\
IM & 16.80 & 12.36 & 7.74 & 6.72 & 5.70 & 6.32 \\
LM$_2$ & 18.72 & 13.18 & 7.84 & 6.78 & 5.70 & 6.32 \\
LM$_3$ & 38.30 & 32.00 & 19.46 & 17.66 & 13.02 & 10.20 \\
\midrule
\multicolumn{7}{l}{Size-adjusted power} \\
\midrule
EM ($K=1$) & 5.22 & 5.18 & 6.84 & 7.46 & 9.40 & 13.02 \\
LRT ($A = [7/9, 2]$) & 7.68 & 7.56 & 8.76 & 8.82 & 9.62 & 12.20 \\
LRT ($A = [2/3, 3]$) & 7.18 & 7.06 & 7.22 & 7.34 & 9.48 & 11.30 \\
LRT ($A = [5/9, 4]$) & 7.10 & 6.22 & 6.90 & 7.08 & 8.20 & 10.54 \\
IM & 5.44 & 4.86 & 4.40 & 4.26 & 4.34 & 5.10 \\
LM$_2$ & 5.44 & 4.86 & 4.40 & 4.26 & 4.34 & 5.10 \\
LM$_3$ & 7.10 & 7.28 & 5.24 & 4.58 & 3.34 & 3.64 \\
\bottomrule
\end{tabular}
\begin{flushleft}
\small
Notes: The nominal level is 5\%, and the number of Monte Carlo replications is 5{,}000. The panel ``Size-adjusted power'' uses the 95th percentile of each test statistic's empirical null distribution in Table \ref{table level independent}. The panel ``Size-unadjusted power'' uses asymptotic critical values, bootstrapped for the LRT.
\end{flushleft}
\end{table}

\begin{table}[h]
\centering
\caption{Covariate-dependent censoring: empirical rejection rates of tests under the null}
\label{table level dependent}
\small
\begin{tabular}{lrrrrrr}
\toprule
$n$ & 50 & 100 & 500 & 1000 & 2000 & 5000 \\
\midrule
EM ($K=1$) & 13.38 & 9.42 & 5.48 & 4.64 & 4.78 & 4.58 \\
EM ($K=2$) & 14.58 & 9.52 & 5.48 & 4.66 & 4.78 & 4.58 \\
EM ($K=3$) & 14.94 & 9.66 & 5.48 & 4.66 & 4.78 & 4.58 \\
LRT ($A = [7/9, 2]$) & 0.00 & 0.00 & 0.74 & 2.14 & 3.02 & 3.72 \\
LRT ($A = [2/3, 3]$) & 0.06 & 0.40 & 2.72 & 3.78 & 3.68 & 3.86 \\
LRT ($A = [5/9, 4]$) & 0.58 & 2.04 & 3.74 & 4.28 & 3.94 & 4.16 \\
IM & 24.28 & 17.90 & 10.10 & 8.24 & 6.98 & 5.72 \\
LM$_2$ & 26.02 & 18.80 & 10.28 & 8.34 & 7.02 & 5.74 \\
LM$_3$ & 61.60 & 47.92 & 29.06 & 23.60 & 18.18 & 13.46 \\
\bottomrule
\end{tabular}
\begin{flushleft}
\small
Notes: Each entry reports the empirical rejection rate of the corresponding test under the null hypothesis. The nominal level is 5\%. The weighted bootstrap sample size for the LRT is 500. The number of Monte Carlo replications is 5{,}000.
The 95th percentiles of the empirical distributions of the EM test statistics when $n = 100$ for $K = 1, 2$, and $3$ are 4.082, 4.106, and 4.106, respectively.
The 99th percentiles of the same distributions for $K = 1, 2$, and $3$ are 7.862, 7.902, and 7.907, respectively.
The 95th and 99th percentiles of the asymptotic null distribution $(\max \{ 0, N(0, 1) \})^2$ are 2.706 and 5.412, respectively.
\end{flushleft}
\end{table}

\begin{table}[h]
\centering
\caption{Covariate-dependent censoring: power of tests for (I) discrete mixture 1}
\label{table power dependent discrete 1}
\small
\begin{tabular}{lrrrrrr}
\toprule
$n$ & 50 & 100 & 500 & 1000 & 2000 & 5000 \\
\midrule
\multicolumn{7}{l}{Size-unadjusted power} \\
\midrule
EM ($K=1$) & 19.60 & 19.42 & 43.38 & 63.94 & 88.48 & 99.62 \\
LRT ($A = [7/9, 2]$) & 0.00 & 0.04 & 16.60 & 47.36 & 83.56 & 99.68 \\
LRT ($A = [2/3, 3]$) & 0.24 & 1.94 & 29.16 & 53.52 & 82.66 & 99.50 \\
LRT ($A = [5/9, 4]$) & 1.70 & 6.02 & 30.50 & 51.90 & 80.42 & 99.38 \\
IM & 19.92 & 15.26 & 21.36 & 39.42 & 69.48 & 98.08 \\
LM$_2$ & 21.76 & 15.96 & 21.66 & 39.50 & 69.58 & 98.08 \\
LM$_3$ & 57.00 & 43.74 & 36.04 & 47.42 & 74.32 & 98.42 \\
\midrule
\multicolumn{7}{l}{Size-adjusted power} \\
\midrule
EM ($K=1$) & 7.56 & 10.54 & 41.90 & 64.80 & 89.04 & 99.74 \\
LRT ($A = [7/9, 2]$) & 10.90 & 15.44 & 42.86 & 64.46 & 89.38 & 99.78 \\
LRT ($A = [2/3, 3]$) & 10.54 & 13.56 & 38.34 & 59.58 & 86.22 & 99.68 \\
LRT ($A = [5/9, 4]$) & 10.08 & 12.46 & 36.40 & 56.50 & 84.24 & 99.58 \\
IM & 3.28 & 2.32 & 9.34 & 25.66 & 63.12 & 97.68 \\
LM$_2$ & 3.28 & 2.32 & 9.34 & 25.66 & 63.12 & 97.68 \\
LM$_3$ & 3.74 & 2.38 & 1.98 & 2.86 & 23.76 & 90.00 \\
\bottomrule
\end{tabular}
\begin{flushleft}
\small
Notes: The nominal level is 5\%, and the number of Monte Carlo replications is 5{,}000. The panel ``Size-adjusted power'' uses the 95th percentile of each test statistic's empirical null distribution in Table \ref{table level dependent}. The panel ``Size-unadjusted power'' uses asymptotic critical values, bootstrapped for the LRT. Results for the EM test are reported for $K=1$ only, as the power for $K=2$ and $K=3$ is very similar.
\end{flushleft}
\end{table}

\begin{table}[h]
\centering
\caption{Covariate-dependent censoring: power of tests for (II) discrete mixture 2}
\label{table power dependent discrete 2}
\small
\begin{tabular}{lrrrrrr}
\toprule
$n$ & 50 & 100 & 500 & 1000 & 2000 & 5000 \\
\midrule
\multicolumn{7}{l}{Size-unadjusted power} \\
\midrule
EM ($K=1$) & 14.96 & 10.44 & 7.04 & 6.46 & 6.44 & 7.06 \\
LRT ($A = [7/9, 2]$) & 0.00 & 0.00 & 0.76 & 2.64 & 4.76 & 7.56 \\
LRT ($A = [2/3, 3]$) & 0.02 & 0.40 & 3.88 & 5.72 & 6.48 & 9.32 \\
LRT ($A = [5/9, 4]$) & 0.40 & 1.68 & 5.44 & 6.64 & 7.26 & 10.14 \\
IM & 23.52 & 17.12 & 9.92 & 7.22 & 6.44 & 5.18 \\
LM$_2$ & 25.58 & 18.08 & 10.08 & 7.28 & 6.46 & 5.18 \\
LM$_3$ & 60.20 & 47.74 & 31.36 & 26.06 & 21.72 & 18.00 \\
\midrule
\multicolumn{7}{l}{Size-adjusted power} \\
\midrule
EM ($K=1$) & 5.72 & 5.50 & 6.66 & 6.74 & 6.76 & 8.16 \\
LRT ($A = [7/9, 2]$) & 4.30 & 4.12 & 5.80 & 6.38 & 7.52 & 9.52 \\
LRT ($A = [2/3, 3]$) & 3.96 & 3.68 & 6.64 & 7.20 & 8.38 & 11.50 \\
LRT ($A = [5/9, 4]$) & 4.46 & 4.26 & 7.24 & 7.68 & 8.92 & 11.78 \\
IM & 4.58 & 4.42 & 4.66 & 3.90 & 4.90 & 4.12 \\
LM$_2$ & 4.58 & 4.42 & 4.66 & 3.90 & 4.90 & 4.12 \\
LM$_3$ & 3.26 & 3.66 & 5.66 & 5.08 & 6.40 & 7.52 \\
\bottomrule
\end{tabular}
\begin{flushleft}
\small
Notes: The nominal level is 5\%, and the number of Monte Carlo replications is 5{,}000. The panel ``Size-adjusted power'' uses the 95th percentile of each test statistic's empirical null distribution in Table \ref{table level dependent}. The panel ``Size-unadjusted power'' uses asymptotic critical values, bootstrapped for the LRT.
\end{flushleft}
\end{table}

\begin{table}[h]
\centering
\caption{Covariate-dependent censoring: power of tests for (III) gamma mixture}
\label{table power dependent gamma}
\small
\begin{tabular}{lrrrrrr}
\toprule
$n$ & 50 & 100 & 500 & 1000 & 2000 & 5000 \\
\midrule
\multicolumn{7}{l}{Size-unadjusted power} \\
\midrule
EM ($K=1$) & 18.96 & 18.70 & 44.56 & 67.60 & 91.00 & 99.82 \\
LRT ($A = [7/9, 2]$) & 0.04 & 0.06 & 13.30 & 44.92 & 83.72 & 99.62 \\
LRT ($A = [2/3, 3]$) & 0.18 & 1.24 & 25.80 & 54.16 & 85.82 & 99.58 \\
LRT ($A = [5/9, 4]$) & 1.20 & 4.38 & 31.22 & 56.32 & 85.56 & 99.58 \\
IM & 21.50 & 16.78 & 26.70 & 47.22 & 79.48 & 99.36 \\
LM$_2$ & 23.14 & 17.52 & 27.08 & 47.36 & 79.56 & 99.36 \\
LM$_3$ & 56.90 & 43.58 & 30.18 & 42.12 & 71.96 & 98.56 \\
\midrule
\multicolumn{7}{l}{Size-adjusted power} \\
\midrule
EM ($K=1$) & 7.20 & 10.14 & 43.22 & 68.54 & 91.48 & 99.84 \\
LRT ($A = [7/9, 2]$) & 9.46 & 13.18 & 41.42 & 64.12 & 90.38 & 99.74 \\
LRT ($A = [2/3, 3]$) & 8.32 & 11.04 & 37.58 & 61.48 & 89.24 & 99.70 \\
LRT ($A = [5/9, 4]$) & 7.54 & 10.76 & 38.06 & 61.24 & 88.40 & 99.66 \\
IM & 4.18 & 3.42 & 12.16 & 32.30 & 72.92 & 99.14 \\
LM$_2$ & 4.18 & 3.42 & 12.16 & 32.30 & 72.92 & 99.14 \\
LM$_3$ & 2.78 & 2.14 & 1.08 & 1.54 & 21.06 & 90.70 \\
\bottomrule
\end{tabular}
\begin{flushleft}
\small
Notes: The nominal level is 5\%, and the number of Monte Carlo replications is 5{,}000. The panel ``Size-adjusted power'' uses the 95th percentile of each test statistic's empirical null distribution in Table \ref{table level dependent}. The panel ``Size-unadjusted power'' uses asymptotic critical values, bootstrapped for the LRT.
\end{flushleft}
\end{table}

\begin{table}[h]
\centering
\caption{Covariate-dependent censoring: power of tests for (IV) log-normal mixture}
\label{table power dependent lognormal}
\small
\begin{tabular}{lrrrrrr}
\toprule
$n$ & 50 & 100 & 500 & 1000 & 2000 & 5000 \\
\midrule
\multicolumn{7}{l}{Size-unadjusted power} \\
\midrule
EM ($K=1$) & 16.92 & 16.16 & 31.50 & 48.68 & 74.88 & 97.40 \\
LRT ($A = [7/9, 2]$) & 0.00 & 0.00 & 7.64 & 27.34 & 62.32 & 95.84 \\
LRT ($A = [2/3, 3]$) & 0.08 & 0.82 & 17.12 & 36.12 & 64.80 & 95.66 \\
LRT ($A = [5/9, 4]$) & 0.78 & 3.44 & 20.30 & 37.36 & 63.94 & 94.72 \\
IM & 20.30 & 16.62 & 16.02 & 28.10 & 54.46 & 91.60 \\
LM$_2$ & 22.24 & 17.58 & 16.20 & 28.22 & 54.48 & 91.62 \\
LM$_3$ & 56.90 & 43.20 & 25.06 & 29.00 & 48.32 & 87.12 \\
\midrule
\multicolumn{7}{l}{Size-adjusted power} \\
\midrule
EM ($K=1$) & 6.18 & 8.44 & 30.28 & 49.56 & 75.86 & 97.78 \\
LRT ($A = [7/9, 2]$) & 7.84 & 10.82 & 28.96 & 46.04 & 73.16 & 97.34 \\
LRT ($A = [2/3, 3]$) & 6.52 & 9.06 & 25.90 & 42.78 & 70.66 & 96.86 \\
LRT ($A = [5/9, 4]$) & 6.30 & 8.60 & 25.74 & 41.62 & 68.96 & 95.90 \\
IM & 3.60 & 3.46 & 6.56 & 16.22 & 46.98 & 90.16 \\
LM$_2$ & 3.60 & 3.46 & 6.56 & 16.22 & 46.98 & 90.16 \\
LM$_3$ & 2.76 & 2.12 & 1.22 & 1.16 & 7.80 & 59.70 \\
\bottomrule
\end{tabular}
\begin{flushleft}
\small
Notes: The nominal level is 5\%, and the number of Monte Carlo replications is 5{,}000. The panel ``Size-adjusted power'' uses the 95th percentile of each test statistic's empirical null distribution in Table \ref{table level dependent}. The panel ``Size-unadjusted power'' uses asymptotic critical values, bootstrapped for the LRT.
\end{flushleft}
\end{table}

\begin{table}[h]
\centering
\caption{Covariate-dependent censoring: power of tests for (V) uniform mixture 1}
\label{table power dependent uniform 1}
\small
\begin{tabular}{lrrrrrr}
\toprule
$n$ & 50 & 100 & 500 & 1000 & 2000 & 5000 \\
\midrule
\multicolumn{7}{l}{Size-unadjusted power} \\
\midrule
EM ($K=1$) & 20.78 & 24.26 & 67.42 & 90.62 & 99.30 & 100.00 \\
LRT ($A = [7/9, 2]$) & 0.00 & 0.30 & 35.06 & 77.48 & 98.54 & 100.00 \\
LRT ($A = [2/3, 3]$) & 0.46 & 2.82 & 49.26 & 82.66 & 98.68 & 100.00 \\
LRT ($A = [5/9, 4]$) & 1.86 & 7.46 & 54.18 & 82.92 & 98.66 & 100.00 \\
IM & 21.48 & 19.88 & 49.04 & 78.40 & 97.56 & 100.00 \\
LM$_2$ & 23.56 & 20.72 & 49.32 & 78.58 & 97.62 & 100.00 \\
LM$_3$ & 54.98 & 41.26 & 46.70 & 72.42 & 96.06 & 100.00 \\
\midrule
\multicolumn{7}{l}{Size-adjusted power} \\
\midrule
EM ($K=1$) & 7.32 & 12.98 & 66.28 & 90.92 & 99.34 & 100.00 \\
LRT ($A = [7/9, 2]$) & 14.50 & 21.56 & 66.92 & 89.24 & 99.34 & 100.00 \\
LRT ($A = [2/3, 3]$) & 12.66 & 18.18 & 62.16 & 86.84 & 99.08 & 100.00 \\
LRT ($A = [5/9, 4]$) & 11.08 & 16.84 & 61.12 & 85.94 & 98.92 & 100.00 \\
IM & 4.12 & 3.84 & 29.16 & 65.84 & 96.44 & 100.00 \\
LM$_2$ & 4.12 & 3.84 & 29.16 & 65.84 & 96.44 & 100.00 \\
LM$_3$ & 3.80 & 1.66 & 1.16 & 8.76 & 63.64 & 100.00 \\
\bottomrule
\end{tabular}
\begin{flushleft}
\small
Notes: The nominal level is 5\%, and the number of Monte Carlo replications is 5{,}000. The panel ``Size-adjusted power'' uses the 95th percentile of each test statistic's empirical null distribution in Table \ref{table level dependent}. The panel ``Size-unadjusted power'' uses asymptotic critical values, bootstrapped for the LRT.
\end{flushleft}
\end{table}

\begin{table}[h]
\centering
\caption{Covariate-dependent censoring: power of tests for (VI) uniform mixture 2}
\label{table power dependent uniform 2}
\small
\begin{tabular}{lrrrrrr}
\toprule
$n$ & 50 & 100 & 500 & 1000 & 2000 & 5000 \\
\midrule
\multicolumn{7}{l}{Size-unadjusted power} \\
\midrule
EM ($K=1$) & 13.92 & 10.10 & 7.96 & 8.42 & 10.36 & 15.12 \\
LRT ($A = [7/9, 2]$) & 0.00 & 0.02 & 1.72 & 4.20 & 6.58 & 11.36 \\
LRT ($A = [2/3, 3]$) & 0.16 & 0.96 & 3.94 & 6.00 & 7.00 & 10.96 \\
LRT ($A = [5/9, 4]$) & 0.98 & 2.60 & 4.70 & 6.36 & 7.30 & 10.36 \\
IM & 23.20 & 17.04 & 8.14 & 6.78 & 6.04 & 6.98 \\
LM$_2$ & 25.04 & 18.04 & 8.22 & 6.84 & 6.10 & 7.02 \\
LM$_3$ & 61.36 & 47.36 & 25.02 & 19.90 & 13.62 & 10.64 \\
\midrule
\multicolumn{7}{l}{Size-adjusted power} \\
\midrule
EM ($K=1$) & 4.80 & 5.02 & 7.34 & 8.68 & 10.84 & 16.60 \\
LRT ($A = [7/9, 2]$) & 6.12 & 7.12 & 8.48 & 9.20 & 10.36 & 14.38 \\
LRT ($A = [2/3, 3]$) & 6.44 & 6.32 & 6.90 & 7.86 & 9.58 & 14.02 \\
LRT ($A = [5/9, 4]$) & 6.00 & 5.88 & 6.86 & 7.70 & 9.36 & 12.76 \\
IM & 5.30 & 4.74 & 3.58 & 3.60 & 4.18 & 5.74 \\
LM$_2$ & 5.30 & 4.74 & 3.58 & 3.60 & 4.18 & 5.74 \\
LM$_3$ & 6.80 & 6.48 & 4.04 & 3.36 & 2.90 & 2.94 \\
\bottomrule
\end{tabular}
\begin{flushleft}
\small
Notes: The nominal level is 5\%, and the number of Monte Carlo replications is 5{,}000. The panel ``Size-adjusted power'' uses the 95th percentile of each test statistic's empirical null distribution in Table \ref{table level dependent}. The panel ``Size-unadjusted power'' uses asymptotic critical values, bootstrapped for the LRT.
\end{flushleft}
\end{table}

\begin{table}[h]
\centering
\caption{The $p$-values for the real-world data analysis}
\label{real-world data analysis}
\small
\begin{tabular}{lccccc}
\toprule
Specification & (I) & (II) & (III) & (IV) & (V) \\
\midrule
\multicolumn{6}{l}{Asymptotic $p$-values} \\
\midrule
EM ($K=1$) & 0.000 & 0.003 & 0.000 & 0.000 & 0.002 \\
EM ($K=2$) & 0.000 & 0.003 & 0.000 & 0.000 & 0.002 \\
EM ($K=3$) & 0.000 & 0.003 & 0.000 & 0.000 & 0.002 \\
LRT ($A = [7/9, 2]$) & 0.142 & 0.326 & 0.078 & 0.106 & 0.110 \\
LRT ($A = [2/3, 3]$) & 0.070 & 0.258 & 0.020 & 0.030 & 0.028 \\
LRT ($A = [5/9, 4]$) & 0.044 & 0.884 & 0.002 & 0.012 & 0.006 \\
IM & 0.023 & 0.057 & 0.002 & 0.005 & 0.004 \\
LM$_2$ & 0.020 & 0.053 & 0.001 & 0.004 & 0.003 \\
LM$_3$ & 0.000 & 0.001 & 0.000 & 0.000 & 0.000 \\
\midrule
\multicolumn{6}{l}{Simulation-based $p$-values ($n = 100$)} \\
\midrule
EM ($K=1$) & 0.002 & 0.013 & 0.001 & 0.002 & 0.009 \\
EM ($K=2$) & 0.002 & 0.013 & 0.001 & 0.002 & 0.009 \\
EM ($K=3$) & 0.002 & 0.013 & 0.001 & 0.002 & 0.009 \\
LRT ($A = [7/9, 2]$) & --- & --- & --- & --- & --- \\
LRT ($A = [2/3, 3]$) & --- & --- & --- & --- & --- \\
LRT ($A = [5/9, 4]$) & --- & --- & --- & --- & --- \\
IM & 0.126 & 0.193 & 0.036 & 0.057 & 0.052 \\
LM$_2$ & 0.126 & 0.193 & 0.036 & 0.057 & 0.052 \\
LM$_3$ & 0.092 & 0.226 & 0.048 & 0.142 & 0.194 \\
\bottomrule
\end{tabular}
\begin{flushleft}
\small
Notes: Each entry reports the $p$-value of the corresponding test statistic for the corresponding specification in \eqref{specification}.
The panel ``Asymptotic $p$-values'' uses the asymptotic null distribution of each test statistic.
The panel ``Simulation-based $p$-values'' instead uses the empirical distribution of each test statistic under the null setting with $n = 100$ reported in Table \ref{table level dependent}; it is a sensitivity analysis rather than a finite-sample correction for these data (see Section \ref{sec:realdata}).
Simulation-based $p$-values are not reported for the LRT, whose asymptotic null distribution is specific to the data-generating process.
\end{flushleft}
\end{table}

\end{document}